\documentclass[11pt]{llncs}
\usepackage{graphicx}
\usepackage{color}

\graphicspath{{./figs/}}
\usepackage{wrapfig}
\usepackage[all]{xy}

\usepackage{amsmath,amssymb,amsbsy,amsfonts,mathrsfs}
\usepackage{euscript,stmaryrd}
\usepackage{enumerate}

\newcommand{\N}{{\mathbb{N}}}

\newcommand{\Lang}{{\cal L}}

\newcommand{\A}{{\tt A}}
\newcommand{\B}{{\tt B}}
\newcommand{\C}{{\tt C}}
\newcommand{\F}{{\tt F}}
\newcommand{\K}{{\tt K}}

\makeatletter
\renewcommand*\l@author[2]{}
\renewcommand*\l@title[2]{}
\makeatletter

\begin{document}

\title{The genus of regular languages}
\author{Guillaume $\text{Bonfante}^1$, Florian $\text{Deloup}^2$}
\institute{Universit\'e de Lorraine, LORIA, Nancy, France
\and Universit\'e Paul Sabatier, IMT, Toulouse, France }

\date{}

 \maketitle

\tableofcontents

\begin{abstract}
The article defines and studies the genus of finite state deterministic automata
(FSA) and regular languages. Indeed, a FSA can be seen as a graph for
which the notion of genus arises. At the same time, a FSA has a
semantics via its underlying language. It is then natural to
make a connection between the languages and the notion of genus. After we introduce and
justify the the notion of the genus for regular languages, the
following questions are addressed. First, depending on the size of the alphabet, we provide upper and lower bounds on the genus of regular languages : we show that under a relatively generic condition on the alphabet and the geometry of the automata, the genus grows at least linearly in terms of the size of the automata. Second, we show that the topological cost of the powerset determinization procedure is exponential. Third, we prove that the notion of minimization is orthogonal to the notion of genus. Fourth, we build regular languages of arbitrary large genus: the notion of genus defines a proper hierarchy of regular languages.
\end{abstract}

Beyond the set-theoretic description of graphs, there is the notion of an
embedding of a graph in a surface. Intuitively speaking, an
embedding of a graph in a surface is a drawing without
edge-crossings. Planar graphs are drawn on the sphere $S_0 $, the
graphs $K_5$ and $K_{3,3}$ are drawn on the torus $S_1$ and more
generally, any graph can be drawn on some closed orientable
surface $S_k$, that is a sphere with $k$ ``handles". The genus of a
graph $G$ is the minimal index $k$ such that $G$ can be drawn on
$S_k$.

\begin{figure}
\begin{center}
\vspace{-4ex}\includegraphics{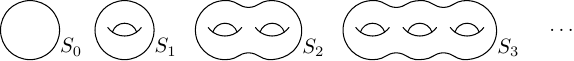}
\end{center}
\end{figure}

\vspace{-4ex}

The aim of this work is to explore standard notions
of finite state automata (FSA) theory with this topological point
of view. The novelty of this point of view lies in the fact that finite
state automata are not only \emph{graphs}, they are
\emph{machines}. These machines compute regular languages. The correspondence is onto: one language
may be computed by infinitely many automata. It is then natural to define the genus of a regular language to be the minimal genus of its representing deterministic automata.

It should be noted that the word ``deterministic" in the previous sentence is crucial: any regular language is recognized by some {\emph{planar}} nondeterministic automaton. The earliest reference for this result we could find is \cite{BoCh}. The cost in terms in extra states and transitions is analyzed in \cite{BezPal}.  By contrast, we show in this paper the existence of regular languages having {\emph{arbitrary high genus}}.

The use of topology in the study of languages may come as a surprise at first. We suggest two
motivations of very different nature. First, the question arises naturally if one wants to build physically the FSA. Think of boolean circuits, they also are graph-machines. There is an immense literature about their electronic implementation, that is about the layout of Very-Large Scale Integration (VLSI) (for instance~\cite{Chen}). In particular, the problem of minimization of  via is  close to the current one. Many contributions suppose a fixed number of layers (holes), but some consider an arbitrary one~\cite{Stallmann}. As we will show, a smaller number of states
may not necessarily mean a smaller cost in terms of the electronic implementation.

There is a second and more fundamental reason why one should consider topology in general and the genus in particular in the study of regular languages. Low-dimensional topology is a natural tool in order to estimate the complexity of languages (or the complexity of the computation of languages).
The main invariant of a regular language $L$ is usually the number of states (the size) of the minimal automaton recognizing $L$. This invariant describes the size of the table data in which transitions are stored, that is the size of the machine's memory. However, simple counting costs memory without complexifying the internal structure of automata. As a simple example, the language $L_{n} = \{ a^{n} \}$ is represented by an automaton of size $n+2$ but with the simple shape of a line:
\begin{center}
\includegraphics[scale=0.7]{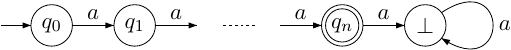}
\end{center}

The genus, as a complexity measure, has been introduced for formal logical proofs by R.~Statman~\cite{Statman}, and further studied by A.~Carbone~\cite{Carbone}.
Cut-elimination is presented as a way of diminishing the complexity of proofs, that is of simplifying proofs. We are not aware of other use of low dimensional topology as a complexity measure besides this work.  To the best of our knowledge, classical textbooks (e.g.~\cite{Hopcroft}, \cite{Sakarovitch}, \cite{Salomaa}) about automata theory are devoted to the set-theoretic approach. Our long term objective is a topological study
of the well known constructions such as minimization, determinization, union, concatenation,
and so on.  This paper is devoted to the notion of genus.

As a first step, we derive a closed formula for the genus of a deterministic
finite automaton (Theorem \ref{th:genus-formula}). Then we show that under a rather mild hypotheses
on the size of the alphabet ($\geq 4$) and on the geometry, the genus of a deterministic finite automaton at least increases linearly in
terms of the number of states (Theorem \ref{th:genus-growth-theorem}). Since the hypotheses depend only on the abstract representing automaton and not on a particular embedding, we deduce an estimation of the genus of regular languages (Theorem \ref{th:genus-growth-language}).

\begin{theorem}
Let $(L_{n})_{n \geq 1}$ be a sequence of regular languages $L_n$ of size~$n$, with alphabet size $m \geq 4$. Assume that for any deterministic automaton recognizing $L_n$, the number of cycles of length $1$ and $2$ is negligible with respect to $n$. For any $\varepsilon > 0$, there is $N > 0$ such that for all $n \geq N$, $$ 1 + \left( \frac{m-3}{6m} - \varepsilon \right) mn \leq g(L_{n}) \leq 1 + \frac{(m-1)n}{2}.$$
\end{theorem}

We present several remarkable consequences of this result throughout this paper.

We mention two particular cases of interest. It is known that the size of the union of two automata increases linearly with the product of their respective size. We prove that
the {\emph{genus}} of the union of two automata $\A$ and $\B$ increases linearly
with the product of the sizes of $\A$ and $\B$ (Corollary \ref{cor:product-automata}). We also provide an example of a nondeterministic automaton $\A$ such that the genus of the powerset-determinized form of a $\A$ is exponential up to a linear factor with respect to the size of $\A$ (Theorem \ref{th:exp-blowup-det}).



In a second step, we study further the link between languages, their representation in terms of automata and their genus. The comparison with state minimization is instructive. Myhill-Nerode Theorem ensures that two deterministic automata with same minimal number of states that recognize the same laguage must be isomorphic. We show that this uniqueness property does not hold if we replace minimal number of states by minimal genus. There is no simple analog to Myhill-Nerode Theorem. As a consequence nonisomorphic automata representing the same language may have minimal genus. There may be even nonisomorphic automata of minimal size within the set of genus-minimal automata.


As a final step, we describe {\emph{explicit}} languages having arbitrary high genus (Theorem \ref{th:hierarchy}). These results imply the existence of a nontrivial hierarchy of regular languages based on the genus and yields a far-reaching generalization of the results of \cite[\S 4]{BoCh}. In particular, the genus yields a nontrivial measure of complexity of regular languages.



\section{Finite State Automata}
We briefly recall the main definitions of the theory of finite state automata and regular languages.
An \emph{alphabet} is a (finite) set of \emph{letters}. A word on an alphabet $A$ is a
finite sequence of letters in the alphabet. Let $A^*$ be the set of all words
on $A$, $\epsilon$ is the empty word and the concatenation of two words $w$ and $w'$ is
denoted by $w \cdot w'$. We define repetitions as follows. Given some word $w$, let $w^0 = \epsilon$ and $w^{n+1} = w^n \cdot w$.

A language on an alphabet $A$ is a subset of $A^*$. Given two languages, let $L+L'$,
$L \cdot L' $ and $L^*$ denote respectively the union, the catenation and the star-operation on $L$ (and $L'$). Rational languages are those languages build from finite sets and the three former operations.

A (finite state) automaton is a $5$-tuple $\A = \langle Q, A, q_0,
F, \delta\rangle$ with $Q$, a finite set of \emph{states} among
which $q_0$ is the \emph{initial} state, $F \subseteq Q$ is the
set of \emph{final} states, $A$ is an alphabet and $\delta
\subseteq Q \times A \times Q$ is the \emph{transition relation}.
The relation $\delta$ extends to words by setting $\delta(q,\epsilon,q)$
for all $q \in Q$ and by defining $\delta(q, a \cdot w, q')$ if and only if
$\delta(q,a,q'') \text{ and } \delta(q'',w,q')$ for some state $q'' \in Q$. Such an automaton
induces a language $$\Lang_\A = \{w \in A^* \mid
\delta(q_0,w,q_f) \wedge q_f \in F\}.$$ The language $\Lang_\A$ is
said to be recognized (or represented) by $\A$. A fundamental result is Kleene's theorem.

\begin{theorem}[Kleene] A language is regular if and only if it is recognized by some finite
state automaton.
\end{theorem}

\begin{example}On the alphabet $\{a,b\}$, let us define the automaton $\F$ on the left and $\K_5$ on the right:
\begin{center}
\includegraphics[scale=0.7]{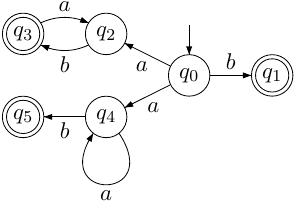} \hspace{2cm} \includegraphics[scale=0.7]{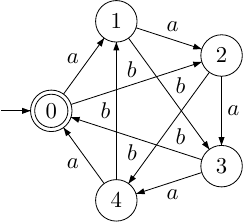}
\end{center}

The small arrows indicate the initial states and final states are doubly circled. The language recognized by $\F$ is $\Lang_\F = \{a^n \cdot b \mid n \in \N\} \cup \{ (a \cdot b)^n \mid n >0\} = a^* \cdot b + (a \cdot b)^*$. The language recognized by $\K_5$ is $\Lang_{\K_5}$ composed of the words of "weight" $0$ modulo $5$.  The weight of $a$ being $1$ and the one of  $b$ being $2$.
\end{example}

An automaton $\A = \langle Q, A, q_0, F, \delta\rangle$ is said to be
\emph{deterministic} (resp. \emph{complete}) if for any state $q \in Q$ and any symbol $a \in A$, the cardinality of the set $\{ q' \in Q \ | \ \delta(q,a,q') \}$ is at most one (resp. at least one). In the case when $\A$ is deterministic and complete, $\delta$ is actually a function  $Q \times A \to Q$. In that case, $\delta(q,u) = q'$ stands for $\delta(q,u,q')$.
It is well known that regular
languages restrict to (complete) deterministic automata. All deterministic automata in this paper shall be finite and complete unless stated otherwise.

\begin{example}The automaton $\K_5$ is deterministic, but $\F$ is not. Nevertheless, $\Lang_\F$ is
recognized by the automaton $\F'$ below:
\begin{center}
\includegraphics[scale=0.7]{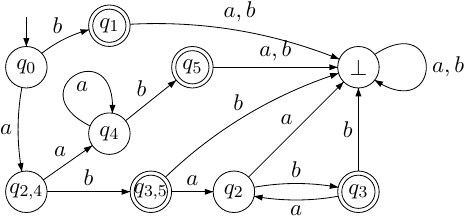}
\end{center}
Note that the only function of the state symbolized by $\perp$ is to make the automaton $\F'$ complete. It is traditionally denoted the ``trash state". Once this state is reached, the final states are unreachable.
\end{example}

Given a language $\Lang$, a distinguishing extension of two words $u$ and $v$ is a word $w$ such that $u\cdot w \in \Lang$ and $v \cdot w \not\in \Lang$. Let $R_{\Lang}$ be the (equivalence) relation $u\; R_{\Lang}\; v$ if and only if $u$ and $v$ have no distinguishing extension.

\begin{theorem}[Myhill-Nerode~\cite{Myhill,Nerode}] A language $\Lang$ is regular if and only if $R_{\Lang}$ has finitely many equivalence classes.
\end{theorem}

Actually, the equivalence classes are the states of an automaton--called the minimal automaton-- which, remarkably, is \emph{the} smallest deterministic automaton recognizing $\Lang$.  By smallest, we mean the one with the minimal number of states. Thus, the notion of size of an automaton $\A$, denoted $|\A|$ in the sequel, is the number of states of $\A$. We emphasized the determinant `the' in the first sentence of the paragraph to stress the fact that there is only one (up to isomorphism) automaton of minimal size representing $\Lang$.

\begin{example}
The automaton $\F'$ is not minimal. But, $\Lang_{\F}$ is recognized by the minimal automaton $\F''$ below.
\begin{center}
\includegraphics[scale=0.7]{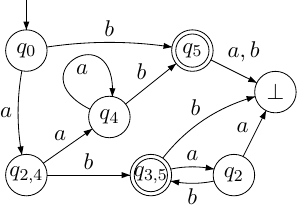}
\end{center}
\end{example}

The size of an automaton serves as an evaluation of its complexity (see for instance~\cite{Yu}).
Due to Myhill-Nerode, one may define the complexity of a regular language to be the size of its
minimal automaton.

\begin{example} There are regular languages of arbitrary large complexity. For instance, on the
alphabet $A = \{ a \}$, consider for all $n >0$ the language $L_n = \{ a^{n} \}$ that consists
of all words on $A$ of length $n$. The linear automaton depicted in the introduction is the
minimal automaton representing $L_n$: it has size $n+2$.
\end{example}

For the proof of Theorem~\ref{th:genus-lower-bound}, we recall the following Proposition (for a proof, see for instance~\cite{Sakarovitch}).

\begin{proposition}\label{pr:projection}
Given two automata $\A = \langle Q, A, q_0, F, \delta\rangle$ and $\A_{min} = \langle Q_{min}, A, q_{0,min}, F_{min}, \delta_{min} \rangle$ representing a common language $\Lang$. Suppose that all states of $\A$  are accessible and that $\A_{min}$ is minimal. Then, there is a function $\rho : Q \to Q_{min}$ such that for all $a \in A$, the following diagram commute:

\begin{center}\hspace{0mm}
\xymatrix{ 1 \ar[r]^{q_0} \ar[rd]_{q_{min}} & Q \ar[rr]^{\delta(-,a)}  \ar[d]_{\rho} && Q \ar[d]^{\rho}\\
&Q_{min} \ar[rr]_{\delta_{min}(-,a)}  &&  Q_{min}}
\end{center}
\end{proposition}


\section{The genus of a regular language} \label{sec:genus-reg-language}

Let $\A$ be a finite automaton. In the constructions to follow, we
regard $\A$ as a graph where the vertices are the states and the
edges are the transitions\footnote{In particular, two vertices may be
joined by several edges.}. We simply forget about the extra
structures on it (namely, the orientation and the labels of the
edges). We are interested in a class of embeddings of $\A$ into
oriented surfaces. Recall that a {\emph{$2$-cell}} is a
topological two-dimensional disc. An automaton is {\emph{planar}}
if it embeds into a $2$-cell (or equivalently a sphere or a
plane).


By means of elementary operations, one can show that $\A$ embeds
into a closed oriented surface $\Sigma$. Among all embeddings that
share that property, choose one such that the complement of the
image of $\A$ in $\Sigma$ is a disjoint union of a finite number
of open $2$-cells. Such an embedding will be called a
{\emph{cellular embedding}}. Again by elementary operations, one
can show that there exists a cellular embedding of $\A$.

As a very simple example, the automaton $\A$ that consists of one
state and one loop embeds in the obvious fashion into the
$2$-sphere. Note that the geometrical realization of $\A$
coincides with the loop.
\begin{wrapfigure}[3]{r}{25mm}
\vspace{-6mm}
\includegraphics{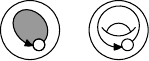}
\end{wrapfigure}
The embedding is cellular because the complement of the loop is
the union of two $2$-cells. The same automaton embeds also into
the torus $T$ as depicted. In this case the embedding is not
cellular because $T \setminus \A$ is a cylinder and not a disjoint
union of $2$-cells.

\begin{example} \label{example:cell-noncell}
Another example is given by an automaton with one state and two
loops. Of the two embeddings depicted into the torus $T$, the top
one is noncellular and the bottom one is cellular. (One should
identify the opposite sides of the square on the left side to
obtain the embedding depicted on the right side.)
\begin{center}
\includegraphics[scale=0.45]{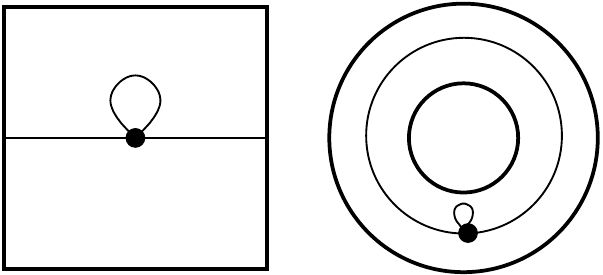} \\
\includegraphics[scale=0.45]{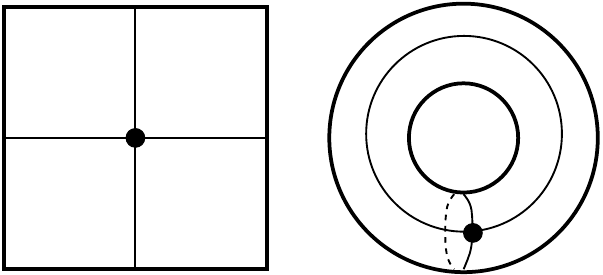}
\end{center}
\end{example}

In this context, the following observation is a tautology:

\begin{lemma} \label{lem:CW-complex}
A cellular embedding of an automaton $\A \subset \Sigma$
determines a finite CW-complex decomposition of the surface
$\Sigma$ in which the $1$-skeleton $\Sigma^{1}$ of $\Sigma$ is
the image of $\A$.
\end{lemma}

A CW-complex is a topological space made up of $k$-dimensional
cells. Here we use $0$-cells (points, corresponding to states),
$1$-cells (topological segments, corresponding to transitions) and
$2$-cells (topological discs). For the precise definition of a
CW-complex decomposition, see for instance \cite[Chap. IV, \S
8]{Bredon}. For instance, the cellular embedding of $\A$ into the
torus $T$ of Example \ref{example:cell-noncell} induces one
CW-complex decomposition of the torus consists of one $0$-cell
(induced by the unique state of $\A$), two $1$-cells (induced by
the two transitions of $\A$) and one $2$-cell (thought of as the
complement of $\A$ in $T$).

Recall that the {\emph{genus}} of a closed oriented surface
$\Sigma$ is the integer $g = \frac{1}{2}\dim
H_{1}(\Sigma;{\mathbb{R}})$. In our context, it is useful to note
that the genus of $\Sigma$ is the maximal number of disjoint
cycles that can be removed from $\Sigma$ such that the complement
remains connected.

\begin{definition}
A cellular embedding of $\A$ into $\Sigma$ is {\emph{minimal}} if
the genus of $\Sigma$ is minimal among all possible surfaces
$\Sigma$ into which $\A$ embeds cellularly.
\end{definition}

\begin{example}
The second embedding of Example \ref{example:cell-noncell} is
cellular: the complement of $\A$ consists in one open $2$-cell. It
is not minimal. Indeed, the automaton embeds into the
$2$-sphere $S^2$: it is realized as the wedge of two circles
$\raisebox{-1mm}{\includegraphics[scale=0.1]{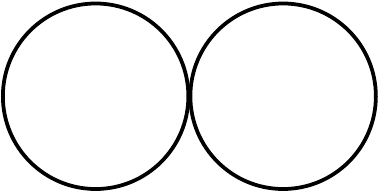}}$
(whose complement in $S^2$ consists of three open $2$-cells).
\end{example}

\begin{definition}
The {\emph{genus}} $g(\A)$ of a finite deterministic automaton $\A$ is
the genus of~$\Sigma$ where $\Sigma$ is a closed oriented surface
into which $\A$ embeds minimally.
\end{definition}

\begin{example}
The genus of the automaton that consists in one state and an arbitrary number
of loops is zero because it embeds into the $2$-sphere.
\end{example}

Let $g_{\A}$ be the smallest number $g_{\Sigma} \in {\mathbb{N}}$
where $\Sigma$ is a closed oriented surface into which $\A$ can be
embedded. Then $g_{\A} \leq g(\A)$ (since all possible embeddings,
included noncellular ones, are considered).

\begin{theorem}[J.W.T. Youngs \cite{Youngs}] \label{th:Youngs} For any automaton $\A$, $g_{\A} = g(\A)$.
In other words, an embedding with minimal genus is cellular.
\end{theorem}

We shall use Youngs' result throughout this paper.

\begin{example} \label{example:K5-in-torus}
Consider the example of the graph $K_{5}$, the complete graph on
five vertices. It is well known that $K_5$ is not planar. Embed it
into the torus $T$ as depicted in Fig. \ref{fig:cellular-emb1}.
\begin{figure}[!h]
\begin{center}
\includegraphics[scale=0.70]{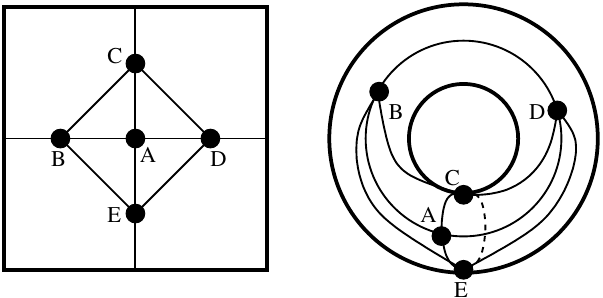}
\end{center}
\caption{A cellular embedding of the graph $K_5$ in the torus
$T$.} \label{fig:cellular-emb1}
\end{figure}
Since the torus has genus $1$, the embedding is minimal. One
verifies that it is also cellular: the complement of $K_5$ in $T$
consists of five disjoint open $2$-cells.

\end{example}

We can now formally state the definition of the genus of a regular
language.

\begin{definition}
Let $L$ be a regular language. The {\emph{genus}} $g(L)$ of $L$ is
the minimal genus of a complete finite deterministic automaton recognizing~$L$:
$$ g(L) = \min \{ g(\A) \ | \ L = L_{\A}, \ \A\ {\hbox{complete finite deterministic}} \}.$$
\end{definition}


There is a simple upper bound for the genus of a deterministic automaton.

\begin{proposition} \label{prop:general-genus-upper-bound}
Let $\A$ be a deterministic automaton with $m$ letters and $n$ states. Then
$$g(\A) \leq 1+ \frac{(m-1)n}{2}.$$
\end{proposition}

\begin{proof}
This follows from Euler's formula (\ref{eq:Euler-Poincare}) (see \S \ref{subsec:preliminary}). \hfill $\blacksquare$ 
\end{proof}

Given some fixed alphabet, Prop. \ref{prop:general-genus-upper-bound} shows that the genus of a regular language $L$ is smaller than the size of a minimal automaton recognizing $L$ up to some linear factor. Hence the main problem we face is to compute a lower bound for the genus (see Th. \ref{th:genus-growth-language}).

The next two results deal with the completeness of automata and reachable states. They are instrumental in nature: they say that the completion of a automaton of minimal genus and the suppression of its unreachable states do not modify the genus. These facts will be used in the sequel without further notice.

\begin{proposition}\label{pr:complete}
For any regular language $L$ with genus $g$, there is a complete, deterministic automaton of genus $g$ representing $L$.
\end{proposition}

\begin{proof}Let $L$ be a regular language with genus $g$. Then, there is a deterministic automaton $\A = \langle Q, A, q_0, F, \delta\rangle$  representing $L$ that embeds cellularly in a surface $\Sigma$ of genus $g(\A) = g$. First, to any state $q$ of $\A$ which would not be complete, add a new trash state $\bot_q$ with the transitions $\delta(q,a) =\bot_q$ for all letter $a$ such that $\delta(q,a)$ is not defined. Second, to each of these new trash states, add loops $\delta(\bot_q,a) = \bot_q$ for all $a \in A$. Clearly, the new transitions embed into $\Sigma$ and do not modify the genus of $\A$.  \hfill $\blacksquare$
\end{proof}

A state $q$ of a (deterministic) automaton $\A = \langle Q, A, q_0, F, \delta\rangle$ is said to be \emph{reachable} if there is a word $w$ such that $\delta(q_0,w) = q$.

\begin{proposition}\label{pr:reachable}For any regular language $L$ with genus $g$, there is a deterministic, complete automaton $\A$ of genus $g$ representing $L$ such that all states of $\A$ are reachable.
\end{proposition}

\begin{proof}Consider an automaton $\A$ of genus $g$ representing $L$. Remove all unreachable states and the corresponding transitions from $\A$. The language recognized by the modified automaton is still $L$. Being a subgraph of $\A$, the new automaton has a genus smaller or equal to $g$.  Since it represents the same language $L$, its genus must be equal to $g$. All its states are reachable. \hfill $\blacksquare$
\end{proof}

\subsection{Combinatorial cycles and faces} \label{subsec:cycles-faces}

In this paragraph, we introduce cycles and faces. A cycle is a notion
that depends only on the abstract graph, while a face depends on a cellular embedding of
the graph. The notion of faces is crucial in the Genus Formula (Theorem \ref{th:genus-formula})
and instrumental in the Genus Growth Theorem (Theorem \ref{th:genus-growth-theorem}).

\begin{definition}
Let $p \geq 1$. A {\emph{walk in $\A$}} is a finite alternating
sequence of vertices (states) and edges $s_0, t_1, s_1, t_2
\ldots, t_p, s_p$ of $\A$ such that for each $j = 1, \ldots, p$,
the states $s_{j-1}$ and $s_{j}$ are the endpoints of the edge
$t_{j}$. The {\emph{length}} of the walk is the number of edges
(counting repetitions). An {\emph{internal vertex}} of the walk is
any vertex in the walk, distinct from the first vertex $s_0$ and
the last vertex $s_p$. The walk is {\emph{closed}} if the first
vertex is the last vertex, $s_0 = s_p$.
\end{definition}

Recall that we regard $\A$ as an unoriented graph: one can walk along an edge
opposite to the original orientation of the transition. The edge
should be nonempty: there should be an actual transition in one
direction or the other. In particular, if there is no transition
from a state $s$ to itself, then the vertex $s$ cannot be repeated
in the sequence defining a walk.

If the underlying graph is simple, then we suppress the notation
of the edges: a walk is represented by a sequence of vertices
$s_0, s_1, \ldots, s_p$ such that any two consecutive vertices are
adjacent.

\begin{definition}
Consider the set $W(p)$ of closed walks of length $p$ in $\A$. The
group of cyclic permutations of $\{1, \ldots, p\}$ acts on $W(p)$.
A {\emph{combinatorial cycle}} of length $p$, or simply a {\emph{$p$-cycle}},
is an orbit of a closed walk of length~$p$.
\end{definition}

In other words, two closed walks represent the same combinatorial cycle
if there is a cyclic permutation that sends one onto the
other. This definition is propped by the fact that we are
interested in geometric cycles only and we do not want to count
them with multiplicities with respect with the start of each node.

\begin{remark}
Our definition of a cycle departs from the traditional one in graph theory:
repetitions of edges and internal vertices may occur. A combinatorial cycle in
which no edge occurs more than one will be called a {\emph{simple}} cycle. (We still
allow repetition of an internal vertex that has several loops.)
\end{remark}

%


We denote the set of all $p$-cycles in $\A$ by
$Z_{p}(\A)$. Since $\A$ is finite,
$Z_{p}(\A)$ is finite. We set $z_{p} = |Z_{p}(\A)|$.


The definitions of walks and cycles are intrinsic to the graph: they do not
depend on an embedding (or a geometric realization) of the graph. However, they
are directly related to topology once an embedding is given.
Let $\A$ be an automaton embedded in a surface $\Sigma$. Each
combinatorial cycle determines a geometric $1$-cycle (in
the sense of singular homology) in $\Sigma$. Therefore combinatorial loops are thought of as
combinatorial analogues of singular $1$-cycles
(in the sense of singular homology).

%

In what follows, consider a cellular embedding of $\A$ into a closed oriented surface $\Sigma$.
By definition, the set $\pi_{0}(\Sigma - \A)$ of connected components of $\Sigma - \A$
consists of a finite number of $2$-cells. The image in $\Sigma$ of the set $\A^{1}$ of
edges of $\A$ is the $1$-skeleton $\Sigma^{1}$ of $\Sigma$. With a slight abuse of notation, we shall
denote by the same symbol $\Sigma^{1}$ the collection of embedded edges of $\A$.
Consider an edge $e \in
\Sigma^{1}$ and an open $2$-cell $c \in \pi_{0}(\Sigma - \A)$. It follows from definitions
that if ${\rm{Int}}(e)$ and ${\rm{Fr}}(c)$ intersect nontrivially then $e \subset {\rm{Fr}}(c)$. Since
$\Sigma$ is a $2$-manifold, there is at most one component $c'$ of $\Sigma - \A$, $c' \not= c$,
such that $e \subset {\rm{Fr}}(c)$.

Without loss of generality, we may assume that the embedded edge $e$ is a smooth arc. Let $x$ be a point in $e$. Define a small nonzero normal vector $\overrightarrow{n}$ at $x$. If $\overrightarrow{n}$ and $-\overrightarrow{n}$ point to distinct components $c, c'$ of $\Sigma - \A$, then $e \subset {\rm{Fr}}(c) \cap {\rm{Fr}}(c')$: there are two distinct components separated by $e$. In this case we say that $e$ is
{\emph{bifacial}}. If $\overrightarrow{n}$ and $-\overrightarrow{n}$ point to the same component $c$ of $\Sigma - \A$, then $c$ is the unique component of $\Sigma - \A$ such that $e \subset {\rm{Fr}}(c)$. In this case, one says that the edge $e$ is {\emph{monofacial}}.

We define
a pairing $\langle -, -\rangle: \Sigma^{1} \times \pi_{0}(\Sigma - \A) \to \{ 0, 1, 2\}$ by

$$ \langle e, c \rangle = \left\{ \begin{array}{cl}
0 & {\rm{if}}\ e \cap {\rm{Fr}}(c) = \varnothing \\
1 & {\rm{if}}\ e \subset {\rm{Fr}}(c)\ {\rm{and}}\ e \subset {\rm{Fr}}(c')\ {\hbox{for}}\ c'\in \pi_{0}(\Sigma - \A) - \{ c \}\\
2 & {\hbox{if $c$ is the unique component of $\Sigma - \A$ such that}}\ e \subset {\rm{Fr}}(c).
 \end{array} \right.$$

From the discussion above, it follows that for any edge $e \in \Sigma^{1}$,
\begin{equation}
\sum_{c \in \pi_{0}(\Sigma - \A)} \langle e, c \rangle = 2. \label{eq:two-factor}
\end{equation}

\begin{definition}
Let $k \geq 1$. A component $c$ of $\Sigma - \A$ is a {\emph{$k$-face}} if $$ \sum_{e \in \Sigma^{1}} \langle e,c \rangle = k.$$
The set of $k$-faces is denoted $F_{k}$. We let $f_{k}$ denote the number of elements in $F_{k}$. A {\emph{face}} is a $k$-face for some $k \geq 1$. The set of faces is denoted $F$.
\end{definition}

The following properties hold:
\begin{enumerate}
\item[(1)] The sets $F_{k}, k \geq 1$ are disjoint;
\item[(2)] All sets $F_{k}$ but finitely many are empty;
\item[(3)] $F = \displaystyle  \bigcup_{k \geq 1} F_{k} =  \pi_{0}(\Sigma - \A)$.
\end{enumerate}

\begin{definition}
Let $k \geq 1$. A {\emph{combinatorial $k$-gon}} in $\Sigma$ is a
$k$-cycle of $\A$ that bounds a $k$-face of $\Sigma - \A$. A $2$-gon will also be called a {\emph{bigon}}. A cycle of length $1$ will be called
a {\emph{loop}}.
\end{definition}

\begin{lemma} \label{lem:one-gon-is-bifacial}
Any $1$-gon has a bifacial edge.
\end{lemma}

The proof follows from the more general fact that a contractible simple closed curve is separating. See \S \ref{subsec:one-gon-is-bifacial} for a proof.


\begin{wrapfigure}[6]{r}{24mm}
\vspace{-0mm}
\includegraphics[scale=0.5]{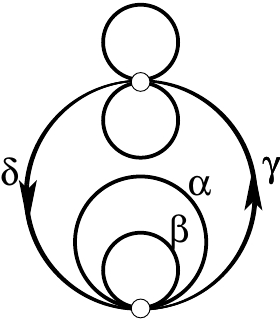}
\end{wrapfigure}

The automaton depicted opposite has two states; each state has three outgoing transitions. It is cellularly embedded into the plane. All edges are bifacial. The edge $\alpha$ is a loop contractible in $\Sigma$ but is not a $1$-gon; $\alpha\beta$ is a cycle of length~$2$ that is a bigon; $\gamma \delta$ is a cycle of length $2$, contractible in $\Sigma$, that is not a bigon.

According to Lemma \ref{lem:one-gon-is-bifacial}, a cycle of length $1$ is monofacial if and only if it is not a $1$-gon (if and only if it represents a nontrivial element in $1$-homology). A bigon may have monofacial edges, even in a cellular embedding: the cellular embedding of Example \ref{example:cell-noncell} provides such an instance.

\begin{remark}
For any $k \geq 1$, $f_{k} \leq z_{k}$. Indeed, just as in
homology, a combinatorial cycle does not necessarily bound a
combinatorial face. For instance, in the $K_5$ embedded in the
torus as in Example~\ref{example:K5-in-torus}, the simple cycle $BCDB$ of
length $3$ is not a $3$-gon. In fact, it does not bound any $2$-cell.
\end{remark}

\begin{lemma}[$1$-gon lemma] \label{lem:one-gon-lemma}
There exists a minimal embedding of $\A$ such that any cycle of length one is a $1$-gon and in particular, is bifacial.
\end{lemma}

See \S \ref{sec:one-gon-lemma} for a proof. In the sequel, we shall frequently use Lemma \ref{lem:one-gon-lemma} without further notice.

\begin{wrapfigure}[3]{r}{10mm}
\vspace{-8mm}
\includegraphics[scale=0.25]{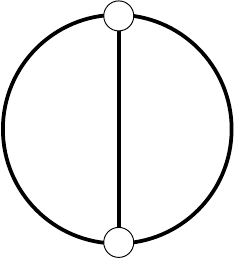}
\end{wrapfigure}

By contrast, there exists automata for which there is no embedding such that every cycle of length two is a bigon. A simple example can be constructed using the subgraph opposite.

\vskip 0.2cm

Our definition of a combinatorial cycle mimics that of a geometric
$1$-cycle $c$ in the sense of a singular $1$-chain such that $\partial
c = 0$. We remark that in order to represent a singular $1$-cycle by
a combinatorial cycle, the combinatorial cycle in question may have
repetitions of edges and internal vertices.

\begin{wrapfigure}[6]{r}{50mm}
\vspace{-7mm}
\includegraphics[scale=0.5]{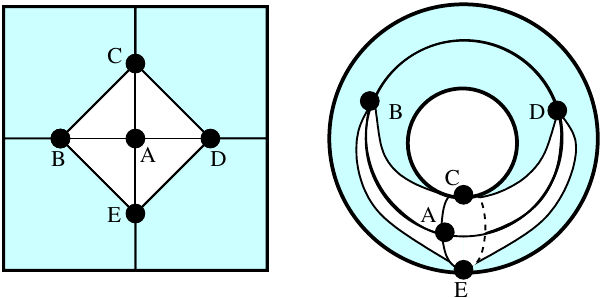}
\end{wrapfigure}


For instance, consider anew the cellular embedding of
Example~\ref{example:K5-in-torus}. As mentioned, the complement of
$K_5$ in $T$ has five components which are open $2$-cells. Four of
them are $3$-faces bounded respectively by $ABCA$, $ACDA$, $ADEA$
and $AEBA$. The fifth component is an open $2$-cell: removing the
other four open $2$-cells and the edges $BD$ and $CE$
yields a $2$-cell. (Removing the four open $2$-cells yields a punctured torus, which is a regular neighborhood of the wedge of a meridian and a longitude; removing the edges $BD$ and $CE$ amounts to cutting transversally the meridian and the longitude respectively, yielding a topological $2$-cell.)

\begin{wrapfigure}[8]{r}{45mm}
\vspace{-5mm}
\begin{center}
\includegraphics[scale=0.3]{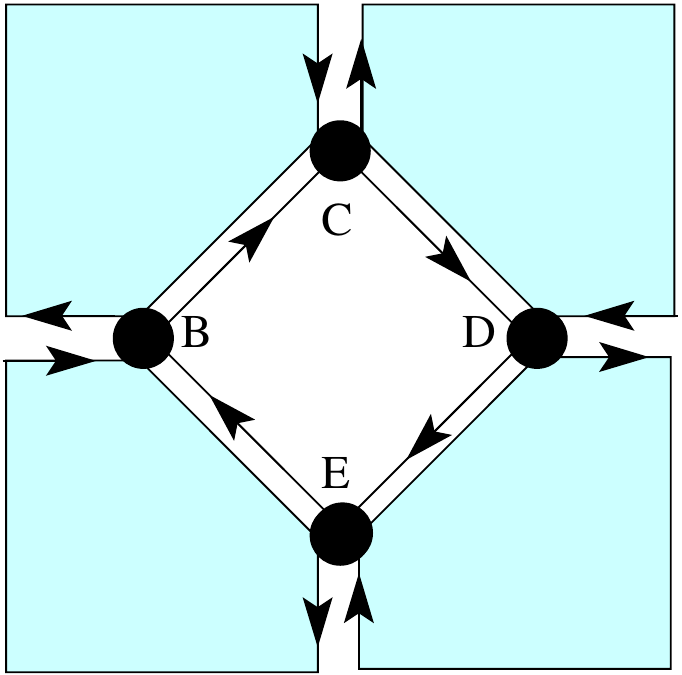}
\end{center}
\end{wrapfigure}
This fifth $2$-cell is a bit more complicated to describe: it is not bounded by any simple cycle. It is bounded by the closed walk
$BCEBDECDB$ which represents a combinatorial cycle of length $8$. It
therefore represents an $8$-gon. Note that the monofacial edges $CE$ and $BD$ are
travelled twice in opposite orientations.

\subsection{Digression: face embeddings and strong face embeddings}

This paragraph is not necessary to understand our results and their proofs (and hence
may be skipped on a first reading). Indeed they do not depend on the notions
introduced here. In particular, they do
not depend on whether the Strong Embedding Conjecture (or a related conjecture)
is true or not. It is true however that for a graph that has a strong embedding in
a surface of minimal genus, then the Genus Formula has a particularly simple form.
(However, it is known in general that this needs not always be the case. There exist
$2$-connected graphs of genus $1$ that have no strong cellular embedding in a torus, see \cite{Xuong}.)
We include this paragraph for clarification.

As we have seen, $k$-gons that appear in cellular embeddings need not be
simple. We may request them to be, at the expense of a more restrictive definition.

\begin{definition}
A {\emph{face embedding}} of a graph $\A$ into a closed oriented
surface $\Sigma$ is a cellular embedding of $\A$ into $\Sigma$
such that each $k$-face in $\Sigma - \A$ is bounded by a simple
$k$-gon.
\end{definition}


\begin{wrapfigure}[5]{r}{30mm}
\vspace{-10mm}
\includegraphics[scale=0.65]{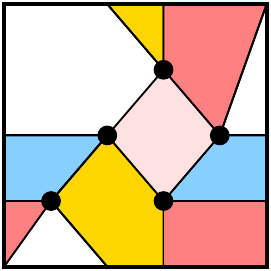}
\end{wrapfigure}
Opposite is depicted another embedding of $K_5$ into the torus
(the opposite sides in the square are identified as usual). This
embedding is a face embedding of $K_{5}$: the complement of $K_5$
consists of five $4$-faces. Hence this embedding is not equivalent
to the cellular embedding of Example \ref{example:K5-in-torus}.

Note that a face embedding does not rule out the possibility that
an edge be monofacial. Recall that $e$ is bifacial if there is a
component $c$ in $\Sigma - \A$ such that $\langle e, c \rangle = 1$.

\begin{definition}
Let $e$ be an edge of embedded graph $\A$ into a closed oriented
surface $\Sigma$.
A {\emph{strong face embedding}} of a graph $\A$ into a closed
oriented surface $\Sigma$ is a face embedding of $\A$ into
$\Sigma$ such that every {\emph{edge}} is bifacial.
\end{definition}

For instance, in Example \ref{example:K5-in-torus}, all but the edges $BD$ and
$CE$ are bifacial. The face embedding of $K_5$ above is
strong. The second embedding of Example \ref{example:cell-noncell}
is a face embedding that is not strong.

This definition is related to that of strong cellular embedding: a
strong cellular embedding of $\A$ is an embedding in $\Sigma$ such
that the closure of each connected component $\Sigma - \A$ is a
closed $2$-cell. Equivalently, every $k$-face of $\Sigma - \A$ is bounded by
a true simple cycle without repetition of an internal vertex.
A strong cellular embedding is a strong face
embedding. The converse does not hold in general.

The strong cellular embedding conjecture \cite{Jaeger} is that every $2$-connected graph
has a strong cellular embedding into some closed surface (orientable or not).
Even though the question is theoretically simpler, we do not know whether
every $2$-connected graph has a strong face embedding into some closed surface
(orientable or not).

\section{Genus Formula}

Our first main result is a closed formula for the genus of a
regular language.

\begin{theorem}[Genus formula] \label{th:genus-formula}
Let $\A$ be a deterministic automaton with $m$ letters. Then for
any cellular embedding of $\A$,
\begin{equation}
g(\A) \leq 1 - \frac{m+1}{4m}f_1 - \frac{1}{2m}f_{2} +
\frac{m-3}{4m}f_{3} + \frac{2m-4}{4m}f_{4} + \frac{3m-5}{4m}f_{5}
+ \cdots
\label{eq:genus-formula}
\end{equation}
with equality if and only if the embedding is minimal.
\end{theorem}

The faces $f_{1}, f_{2}, \ldots$ are determined by the cellular
embedding of $\A$. It follows from \S \ref{subsec:cycles-faces}
that for each cellular embedding, there is some $M > 0$ such that $f_{k} = 0$ for all $k \geq M$.
In particular, the  sum $\sum_{k = 1}^{\infty} \frac{k(m-1) - 2m}{4m}f_{k}$
that appears on the right hand
side of $(\ref{eq:genus-formula})$ is finite.

\begin{remark}
In the case when $(\ref{eq:genus-formula})$ is an equality, it is {\emph{not}} claimed that the embedding is unique.
Thus inequivalent minimal embeddings for $\A$
lead to distinct formulas for the genus of $\A$.
\end{remark}

\begin{remark}
In the case when $(\ref{eq:genus-formula})$ is an equality,
it is {\emph{not}} claimed that the automaton
$\A$ is the minimal state automaton (in the sense
of Myhill-Nerode). Indeed, the automaton with the least number of
states does not have necessarily minimal genus (see below \S
\ref{sec:state-versus-genus}).
\end{remark}

The Genus Formula is proved
in \S \ref{sec:genus-formula-proof}.

\section{Genus growth}

Let us begin with a simple example. Any language on a $1$-letter alphabet is represented by a
planar deterministic automaton. Indeed, these deterministic automata have a quasi-loop (planar) shape:
\begin{center}
\includegraphics[scale=0.8]{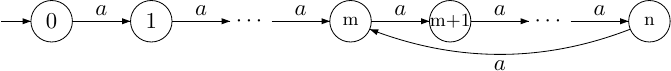}
\end{center}

Actually, there is an indirect proof of the result. For a unary alphabet, $e_1 \leq e_0$.
Since there is at least one face, Euler's relation states that $1 \leq 2 - 2g$, that is $g \leq 1/2$.
And thus $g = 0$. The remark shows that to get graphs of higher genus, one should augment the
number of edges. Then the relation $e_1 = m e_0$ forces to augment the size
of alphabets. Thus a general study of the genus of automata depends on the size of alphabet.
Consider now the case of a $2$-letter alphabet.


\begin{proposition} \label{prop:two-letters}
Let $(\A_n)_{n \in {\mathbb{N}}^{\times}}$ be a sequence of deterministic finite automata of size $n$ on the same alphabet. Assume that the number $m$ of letters is two. For any cellular embedding of $\A_n$,
\begin{enumerate}
\item[$\bullet$] either there exists $M > 0$ such that $\underset{n \geq 1}{\sup}\ \bigl(f_1(n) + f_2(n) + f_3(n)\bigl) \leq M$
\item[$\bullet$] or $\underset{n \to +\infty}{\lim} \sum_{k \geq 5} \frac{k-4}{8} f_{k}(n) = +\infty$.
\end{enumerate}
\end{proposition}

\begin{proof}
Suppose neither condition is satisfied. Then $$f_1(n) \underset{n \to+\infty}{\to} +\infty \ \ {\rm{or}}\ \ f_2(n) \underset{n \to +\infty}{\to} +\infty\ \ {\rm{or}}\ \ f_3(n) \underset{n \to +\infty}{\to} +\infty$$ (since these are sequences of nonnegative integers) and $\sum_{k \geq 5} f_{k}(n)$ remains bounded. For $m = 2$,
the second, third and fourth terms respectively in the genus inequality (\ref{eq:genus-formula}) are negative or zero. The fifth term is always zero for $m = 2$. It follows easily that $g(n)$ is negative for $n$ large enough, which is a contradiction. \hfill $\blacksquare$
\end{proof}

\begin{corollary} \label{cor:two-letters}
 Let $(L_n)_{n \in \mathbb{N}}$ be a sequence of regular languages on two letters. If for each $n$, $L_n$ is recognized by a deterministic automaton $\A_{n}$ of size $n$ having a cellular embedding such that $\sum_{k \geq 5} \frac{k-4}{8} f_k(n)$ remains bounded as $n \to +\infty$, then the genus $g(L_n)$ of $L_n$ also remains bounded.
\end{corollary}
\begin{proof}
By hypothesis, the number $f_{k}(n)$ of $k$ faces in a cellular embedding of $\A_{n}$ into a surface $\Sigma_n$ verify  $\sum_{k \geq 5} \frac{k-4}{8} f_{k}(n) < + \infty$. By Prop. \ref{prop:two-letters},  $f_1(n)$, $f_2(n)$ and $f_3(n)$ are bounded. The genus formula then shows that $g(\A_{n})$ remains bounded. Since $g(L_n) \leq g(\A_{n})$, the conclusion follows. \hfill $\blacksquare$
\end{proof}

Each $\A_n$ in general has several nonequivalent cellular embeddings. But for {\emph{any}} cellular embedding, the alternative of Prop. \ref{prop:two-letters} holds for the various numbers of faces $f_{k}(n)$ (determined by the embedding). Corollary \ref{cor:two-letters} states a sufficient condition for the genus of a {\emph{language}} to be bounded.
The interest in this result lies in the fact that it discriminates between the respective contribution of the faces to the genus.



There is a similar result when the number of letters is three.

\begin{proposition} \label{prop:three-letters}
Let $(\A_n)_{n \in {\mathbb{N}}^{\times}}$ be a sequence of deterministic finite automata of size $n$ on the same alphabet. Assume that the number $m$ of letters is three. For any minimal embedding of $\A$,
\begin{enumerate}
\item[$\bullet$] either there exists $M > 0$ such that $\underset{n \geq 1}{\sup}\ \bigl(f_1(n) + f_2(n) \bigl) \leq M$
\item[$\bullet$] or $\underset{n \to +\infty}{\lim} \sum_{k \geq 4} \frac{k-3}{6} f_{k}(n) = +\infty$.
\end{enumerate}
\end{proposition}

\begin{corollary} \label{cor:three-letters}
 Let $(L_n)_{n \in \mathbb{N}}$ be a sequence of regular languages on two letters. If for each $n$, $L_n$ is recognized by a deterministic automaton $\A_{n}$ of size $n$ having a cellular embedding such that $\sum_{k \geq 4} \frac{k-3}{6} f_k(n)$ remains bounded as $n \to +\infty$, then the genus $g(L_n)$ of $L_n$ also remains bounded.
\end{corollary}

The proofs of Prop. \ref{prop:three-letters} and Cor. \ref{cor:three-letters} are similar to those of Proposition \ref{prop:two-letters} and Corollary \ref{cor:two-letters}.

We state our main result on the genus growth of automata and languages.

\begin{theorem}[Genus Growth] \label{th:genus-growth-theorem}
Let $(\A_{n})_{n \in {\mathbb{N}}^{\times}}$ be a sequence of
deterministic finite automata with $m$ letters and $n \geq 1$
states. Let $g(n)$ be the genus of $\A_{n}$. Assume
\begin{enumerate}
\item[$(1)$] $m \geq 4$. \item[$(2)$] The numbers
$z_{k}(n)$ of cycles of length $1$ and $2$ in $\A_{n}$ are
negligible with respect to the size $n$ of $\A_n$: $\underset{n \to
+\infty}{\lim}\frac{z_{1}(n)}{n} = \underset{n \to
+\infty}{\lim}\frac{z_{2}(n)}{n} = 0$.
\end{enumerate}
Then for any $\varepsilon > 0$, there exists $N > 0$ such that for
all $ n \geq N$,
$$g(n) \geq 1 +  \left( \frac{m-3}{6m} - \varepsilon \right) mn.$$
\end{theorem}

The Genus Growth Theorem is proved in \S \ref{sec:genus-growth-proof}.


We begin with examples showing that we cannot easily dispense with the hypotheses of Theorem \ref{th:genus-growth-theorem}.

\begin{example}[Quasi-loop automaton]
Any quasi-loop finite deterministic automaton with alphabet of cardinality $m \leq 1$ is planar.
This of course does not contradict the Genus Growth Theorem because Hypothesis $(1)$ does not hold.
\end{example}

\begin{example}[Genus $1$ automaton]
Let $n \geq 3$. Define an automaton $\A_{n}$ as follows. Consider the set $\tilde{S}_n = \{ (i,j) \ | \ 0 \leq i,j \leq n \}$ inside the square $C = [0,n] \times [0,n]$. The quotient $T$ of $C$ under the identifications $(0,t) = (1,t)$ and $(t,0) = (t,1)$ is a torus. The image $S_n$ of $\tilde{S}_n$ in $T$ is the set of states. Note that there are exactly $n^2$ states. For each $(i,j) \in (\mathbb{Z}/n\mathbb{Z})^2$, define two outgoing transitions $(i,j) \to (i+1,j)$ (mod $n$) and $(i,j) \to (i,j+1)$ (mod $n$). Choosing arbitrary initial and final states yields a finite deterministic complete automaton $\A_n$ with $n^2$ states. Clearly $g(\A_n) \leq 1$ for any $n \geq 3$. (This also follows from Cor. \ref{cor:two-letters}.) This does not contradict the Genus Growth Theorem because the alphabet has only two letters. The same example with an extra outgoing transition $(i,j) \to (i+1, j+1)$ (i.e., with an extra letter for the alphabet) still yields an automaton $\B_n$ with $g(\B_n) \leq 1$ for any $n \geq 3$. (This also follows from Cor. \ref{cor:three-letters}.) This still does not contradict the Genus Growth Theorem: the alphabet has only three letters.
\end{example}

\begin{example}[Another genus $1$ automaton]
Start with the previous example $\B_n$. To each state $(i,j)$, add an outgoing transition pointing to $(i,j)$. This yields a deterministic complete automaton with $n^2$ states and an alphabet that consists now of $4$ letters. There is an obvious cellular (minimal) embedding in the torus $T$ as before. This does not contradict the Genus Growth Theorem because now the number of cycles of length one (loops) is $n^2$ (the number of states), so the second hypothesis is not satisfied.
\end{example}

\begin{center}
\includegraphics[scale=0.2]{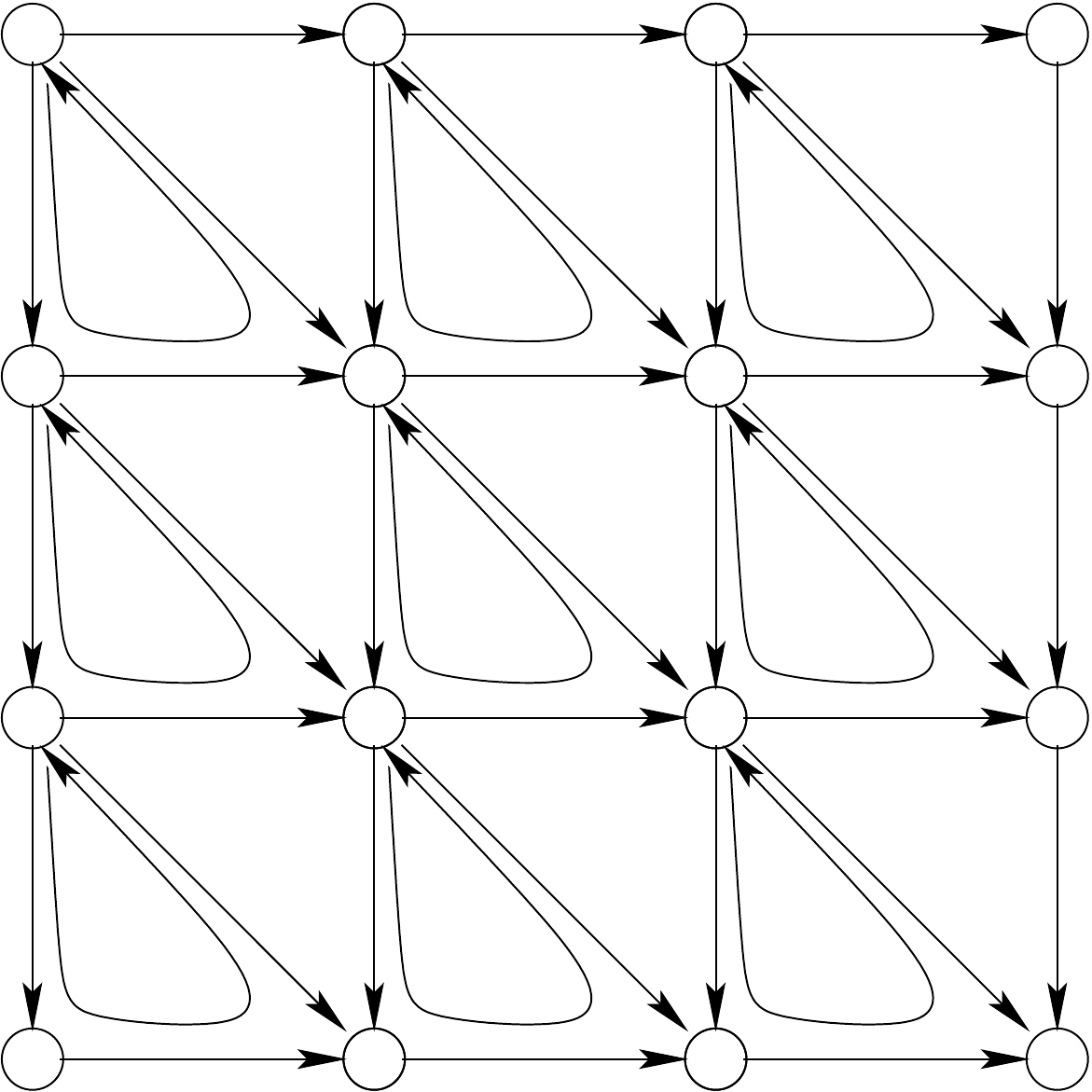}
\end{center}

A natural consequence of the Genus Growth Theorem for automata is an estimation of the genus of regular languages.

\begin{theorem}[Genus Growth of Languages] \label{th:genus-growth-language}
Let $(L_{n})_{n \in \mathbb{N}}$ be a sequence of regular languages on $m$ letters, $m \geq 4$.
Suppose that for each $n$ large enough, any automaton recognizing $L_{n}$ has at least $n$ states
and that the number of its cycles of length $1$ and $2$ are negligible with respect to $n$. Then
for any $\varepsilon > 0$, there is $N > 0$ such that for all $n \geq N$,
$$ 1 +  \left( \frac{m-3}{6m} - \varepsilon \right) mn \leq g(L_n) \leq 1+\frac{(m-1)n}{2}$$
\end{theorem}

\begin{proof}
The upper bound for the genus follows from Prop. \ref{prop:general-genus-upper-bound}. The
lower bound is a direct consequence of the Genus Growth Theorem. \hfill $\blacksquare$
\end{proof}

In particular, under the hypothesis of Theorem \ref{th:genus-growth-language}, the genus $g(L_{n})$
grows linearly in the size $n$ of the minimal automaton $\A_{n}$ representing~$L_n$.


We take up the question of explicitly constructing such a sequence of regular languages
in \S \ref{subsec:hierarchy}. There we detail an explicit construction, that shows that
there is a hierarchy of regular languages based on the genus (Th.~\ref{th:hierarchy}).

Another application of the Genus Growth Theorem is the estimation of the genus of
product automata in \S \ref{subsec:product-automata}.

\begin{theorem}[Genus Lower Bound]~\label{th:genus-lower-bound} Given an alphabet of size $m \geq 4$,
for any regular language $\Lang$, if $\Lang$ is recognized by a
minimal (complete) automaton of size at least $n$ without loops
nor cycles of length $2$, then
$$ 1 +  \left( \frac{m-3}{6} \right) n \leq g(\Lang) \leq 
1+\frac{(m-1)n}{2}$$
\end{theorem}

Compared to Theorem~\ref{th:genus-growth-language}, the genus of the language depends on a condition dealing with only one automaton, namely the minimal automaton. The proof is given is \S \ref{sec:genus-size}.

\subsection{The genus of product automata} \label{subsec:product-automata}

 It is well know that the size of the product automaton corresponding to the union of
 two deterministic automata $\A$ and $\B$ is bounded by $m \times n$, the product of
 the size of $\A$ and the size of $\B$. This bound is actually a lower bound as presented by
 S.~Yu in~\cite{Yu}. By Prop. \ref{prop:general-genus-upper-bound}, up to a linear factor due to
 the size of the alphabet, $m\times n$ is also an upper bound on the genus of the
 product automaton. We prove that it is also a lower bound.

\begin{corollary} \label{cor:product-automata}
There is a family $(\A_m,\B_n)_{m\in\N, n\in \N}$ of  planar automata $\A_m$ and $\B_n$
of respective size $m$ and $n$ such that the deterministic minimal
automata $\A_n \cup \B_m$ has genus $O(m \times n)$.
\end{corollary}

\begin{proof}

Let $\A_m$ be the $m$-state automaton defined as follows.

\vspace{-1ex}
\begin{center}
\includegraphics[scale=0.8]{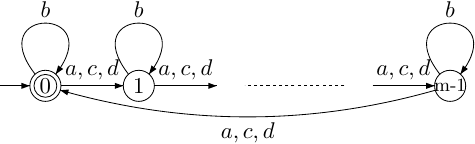}
\end{center}

\vspace{-3ex}
Let $\B_n$ be the $n$-state automaton defined as follows.
\begin{center}
\includegraphics[scale=0.8]{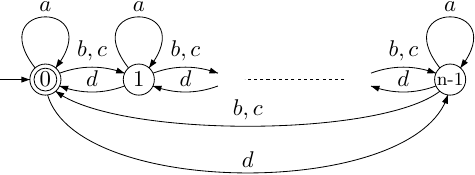}
\end{center}
The minimal automaton $\A_m \cup \B_n$ has size $m \times n$ and it contains neither loops,
nor bigons. Thus, Theorem~\ref{th:genus-growth-theorem} applies and leads to the conclusion. \hfill $\blacksquare$
\end{proof}

\subsection{A hierarchy of regular languages} \label{subsec:hierarchy}

Is there always  a planar deterministic representation of a regular language? For finite languages, the answer is positive. Indeed, finite languages are represented by trees (which are planar). In general, as evidenced by the Genus Growth Theorem for regular languages, the answer is a clear ``no''.  This section is devoted to an explicit constructive proof.

\begin{theorem}[Genus-Based Hierarchy] \label{th:hierarchy}
There are regular languages of arbitrarily large genus.
\end{theorem}

\begin{proof} Consider the alphabet $A = \{a, b, c, d\}$. Consider the
automata $\A_3$ and $B_n$ defined previously (in \S\ref{subsec:product-automata}) of
sizes $3$ and $n$ respectively.

%
The minimal automaton $U_n$ representing the language $\A_3 \cup \B_n$ has size $3 \times n$ (see the figure below\footnote{To avoid an inextricable drawing, we made some transitions point to a shadow of their target. Alternatively, view the drawing (with crossings) on a torus.}) and it contains neither loops nor bigons whenever $n \geq 3$. Thus Theorem~\ref{th:genus-lower-bound} applies and achieves the proof. \hfill $\blacksquare$
\end{proof}

\begin{center}
\includegraphics[scale=0.7]{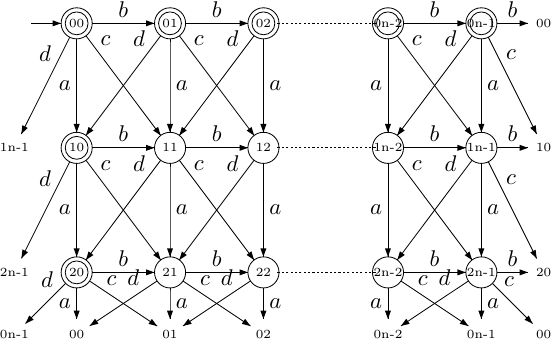}
\end{center}

\subsection{The exponential genus growth of determinization}

We prove that determinization leads to an exponential genus growth, as
this is the case for state-complexity (see for
instance~\cite{GaoYu}). Consider the following family of  automata
$(\A_n)_{n \in \N^\times}$. The alphabet $\mathscr{A}_n = \{ x_1,
\ldots, x_n\}$ is a set of cardinality $n$. The states of $\A_n$
consist of one initial state $s_0$, $n$ states (one state for each
letter) $s_1, \ldots, s_n$ and one trash state. All states except
the initial state and the trash state are final. The transitions
of $\A_n$ are defined as follows:
\begin{enumerate}
\item[$\bullet$] From the initial state $s_0$ to each state $s_i$
($1 \leq i \leq n$), there are $n-1$ transitions whose labels lie
in $\mathscr{A}_n - \{ x_i \}$. \item[$\bullet$] From each state $s_i$
($1 \leq i \leq n$) to itself, there are $n-1$ transitions whose
labels lie in $\mathscr{A}_n - \{ x_i \}$. \item[$\bullet$] From each
state $s_i$ ($1 \leq i \leq n$) to the trash state, there is one
transition whose label is $x_i$. (One can add $n$ loops with
labels $x_1, \ldots, x_n$ to the trash state so that the resulting
automaton is complete.)
\end{enumerate}

If follows from the definition that the language recognized by
$\A_n$ is the set of words containing at most $n-1$ distinct
letters.  It is also clear from the definition that for any $n
\geq 2$, $\A_n$ is planar and nondeterministic. (Note that the
fact that we include or not the trash state with or with its loops
is irrelevant.)

\begin{theorem} \label{th:exp-blowup-det}
The determinization of $\A_n$ is minimal and has genus
$$g_n \geq 1 + \left(\frac{n}{4} -1\right) 2^{n-1}.$$
\end{theorem}

For instance, $g_4 \geq 1$ so the determinization
$\A_{4}^{\rm{det}}$ of $\A_{4}$ is not planar. This can be seen by
Kuratowski's theorem (as it is can be seen $\A_{4}^{\rm{det}}$
contains the utility graph $K_{3,3}$). It is not hard to embed
$\A_{3}^{\rm{det}}$ into a plane so $g_3 = 0$. Of course the
meaning of the theorem is that the genus of $\A_{n}^{\rm{det}}$
grows at least exponentially in $n$.

\begin{figure}[!ht]
\raisebox{5mm}{\includegraphics[width=6cm]{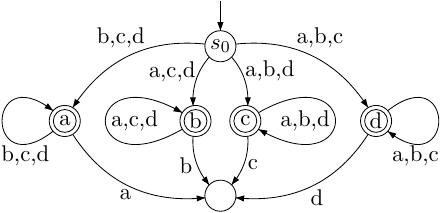}}
\includegraphics[width=6cm]{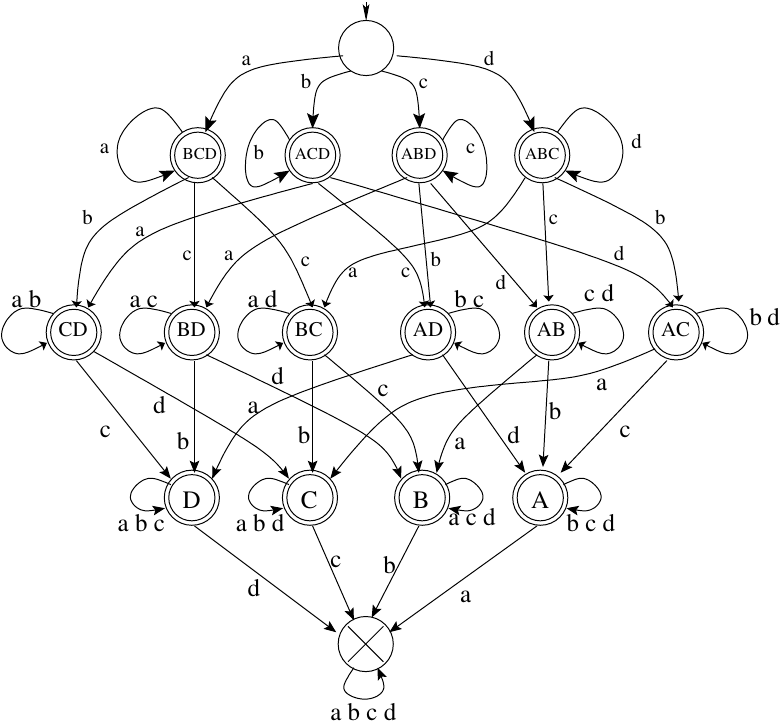}
\caption{The automaton $\A_4$ and its determinized form.}
\label{fig:det-auto}
\end{figure}%

\begin{proof}

Let us describe an isomorphic variant $\A^{\rm{det}}_n$ of the
determinized form of $\A$ by the powerset method.  The states of
$\A^{\rm{det}}_n$ consist of all subsets of $\Sigma_n$. The
initial state of $\A^{\rm{det}}_n$ is ${\mathscr{A}}_n$ itself. The trash
state is the empty set. Any state but the trash state  and the
initial state is a final state.

Therefore, the number $e_{0}$ of states of $\A^{\rm{det}}_n$ is $2^n$. 
The transitions are described as follows. For each letter $x \in {\mathscr{A}}_n$, there is one transition from $S$ to the state $S - \{x\}$. 
Minimality follows from definitions: there are no
indistinguishable states.

Let us consider the number $e_{1}^{o}$ of transitions of
$A^{\rm{det}}$ that are loops. By definition, each state labelled
by a subset of cardinality $k$ contributes exactly $n-k$ loops. We
conclude that
\begin{equation}
 e_{1}^{o} = \sum_{k = 0}^{n} {n\choose k} \ (n-k) = \sum_{k=0}^{n} {n \choose k}\ k = n\cdot 2^{n-1}.
 \label{eq:number-loops}
\end{equation}

It follows that exactly half of the transitions are loops:

\begin{equation}
e_{1} = 2e_{1}^{o}. \label{eq:transition-loop}
\end{equation}

Consider now a minimal embedding of $\A_{n}^{\rm{det}}$ into a closed oriented surface $\Sigma$.
Consider one loop $l$ in $\Sigma^{1}$. Since it is bifacial, it is the intersection
of exactly two distinct adjacent closed $2$-cells. Therefore removing the loop (while
keeping the state) amounts to merging two $2$-cells into one
$2$-cell. The union of states and transitions (minus $l$) still
induces a $CW$-complex decomposition of $\Sigma$.
Therefore, according to Euler's relation, the genus of $\Sigma$ is
unaffected. We can therefore remove all loops from $\Sigma^{1}$. Thus we
can assume that $e_1 = e_{1}^{o}$ (from
(\ref{eq:transition-loop})) and $f_1 = 0$.

\begin{lemma} \label{lem:properties}
For the new graph minimally embedded in $\Sigma$, the following properties hold:
\begin{enumerate}
\item[$\bullet$] $f_2 = 0$; \item[$\bullet$] For any $k \geq 1$,
$f_{2k+1} = 0$; \item[$\bullet$] For any $k \geq n$, $f_{2k} = 0$.
\end{enumerate}
\end{lemma}

\noindent{\emph{Proof}}. These observations are consequences of the particular structure
of the original graph $\A_{n}^{\rm{det}}$:
they follow from the definition of $A^{\rm{det}}$ and are left to the
reader. \hfill $\blacksquare$\\

\noindent We return to the proof of Theorem \ref{th:exp-blowup-det}. We have
$$ 2e_1 = f_1 + 2f_2 + 3f_3 + 4f_4 \cdots = 4f_4 + 6f_6 + \cdots + (2n-2)f_{2n-2}.$$
The first equality is relation (\ref{eq:degree-formula}) and the
second equality follows from Lemma \ref{lem:properties}. Since all
numbers are nonnegative numbers, we have
$$ 2e_1 \geq 4 (f_4 + f_6 + \cdots + f_{2n-2}) = 4\ e_2.$$
Thus $e_2 \leq \frac{1}{2} e_1$. From Euler's relation, we deduce
that
$$ 2g = 2 - e_0 + e_1 - e_2 \geq 2 - e_0 + e_1 - \frac{1}{2} e_1 = 2 - e_0 + \frac{1}{2} e_1.$$
Substituting values for $e_0$ and $e_1$, we obtain
$$ 2g \geq 2 + \left(\frac{n}{4} - 1\right)2^{n}.$$
This is the desired result. \hfill $\blacksquare$

\end{proof}

\section{State-minimal automata versus genus-minimal automata} \label{sec:state-versus-genus}

Minimal automata--as given by Myhill-Nerode Theorem--have the remarkable properties
to be unique up to isomorphism, leading to a fruitful relation between rational languages
and automata. In this section, we show that state-minimality is a notion orthogonal to genus-minimality.
First, consider the following proposition:

\begin{proposition} There are deterministic automata with a genus strictly lower than the genus of
their corresponding minimal automaton.
\end{proposition}

\begin{proof}Let  $\K_5$, $\K'$ be the automata:

\vspace{1ex}\includegraphics[scale=0.8]{genus-pics-4.pdf}
\hspace{2cm}\includegraphics[scale=0.8]{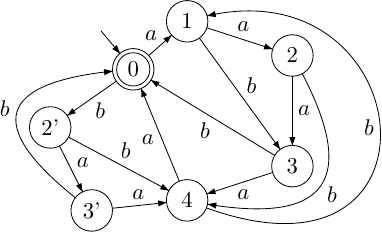}

Clearly, $\K_5$ and $\K'$ represent the same language $\Lang$, $\K_5$ is minimal, $\K_5$ has genus $1$ and $\K'$ is planar.
\end{proof}

\begin{example}Minimal automata need not have maximal genus.
For instance, the following deterministic automaton
\begin{center}
\includegraphics[scale=0.8]{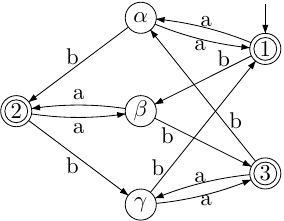}
\end{center}
has genus 1, but its corresponding minimal automaton
\raisebox{-2mm}{\includegraphics[scale=0.7]{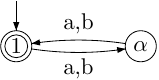}}
has genus 0.
\end{example}

\begin{wrapfigure}[6]{r}{30mm}
\vspace{-9.5mm}
\includegraphics[scale=0.65]{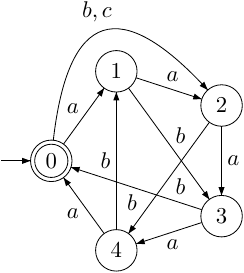}
\end{wrapfigure}

Contrarily to set-theoretic state minimization (Myhill-Nerode
theorem), there is no isomorphism between genus-minimal automata,
even within the class of genus-minimal automata having minimal
state size. Consider the language of the minimal automaton
$\K_5^*$ opposite. It is represented by the two automata
$\K_{5,1}^*$ (middle) and $\K_{5,2}^*$ (right) below. (To save
space we have omitted all loops based at each state $i$, $0 < i
\leq 4$, with label $c$.)

\vspace{2ex}

\hspace{0.4cm} \includegraphics[scale=0.6]{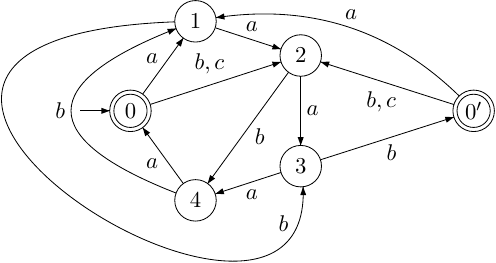}
\hspace{0.4cm} \includegraphics[scale=0.6]{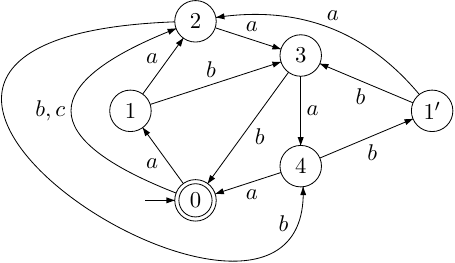}

\vspace{2ex} \noindent $\K_{5,1}^*$ and $\K_{5,2}^*$ are both
planar, thus have minimal genus. They have minimal size within the
set of automata of minimal genus. Indeed, there is only one
automaton with a lower number of states, it is $\K_5^*$, and it is
not planar. Finally, $\K_{5,1}^*$ and $\K_{5,2}^*$ are not
isomorphic: $\K_{5,1}^*$ contains two non-trivial edges (that are
not loops) labelled $c$ where $\K_{5,2}^*$ contains only one. To
sum up, the two automata $\K_{5,1}^*$ and $\K_{5,2}^*$  (a)
represent the same language, (b) have minimal genus, (c) have
minimal size given that genus, (d) have non isomorphic underlying
graphs.

\section{Nondeterministic planar representation}

The genus of a regular language $L$ was defined in \S \ref{sec:genus-reg-language} as the minimal genus of a {\emph{deterministic}} automaton recognizing $L$. In this section, we point out that that the word ``deterministic" is crucial in the previous sentence. The following result is essentially proved by R.V. Book and A.K. Chandra \cite[Th. 1a \& 1b]{BoCh}. (See also \cite{BezPal}.)

\begin{theorem}[Planar Nondeterministic Representation]
For any regular language $L$, there exists a planar nondeterministic automaton $\A$ recognizing $L$.
\end{theorem}

\begin{proof}
 We include two proofs for the convenience of the reader. Both follow closely \cite{BoCh} with minor modifications. Let $L = L(R)$ be a regular language given by a regular expression $R$. We shall show that
$L = L(\A)$ for some planar nondeterministic automaton $\A$.

The proof follows the recursive definition of a regular expression. An expression that is not the empty string is regular if and only if it is constructed from a finite alphabet using the operations of union, concatenation and Kleene's ${\hphantom{a}}^+$-operation.
Consider the class ${\mathcal{C}}$ of planar finite nondeterministic automata that have exactly one initial state, exactly one final state such that the initial state and the final state are distinct.

Clearly, $\mathcal{C}$ contains an automaton that recognizes the regular expressions $R = \varnothing$ (take $\A$ to be the automaton with two states, one initial, one final and no transition) and $R = a \in A$ (take the automaton with two states, one initial, one final and one $a$-labelled transition from the initial state to the final state).

Next, we show that the class ${\mathcal{C}}$ is closed under the three operations mentioned above. Suppose given two subexpressions $R$ and $S$ recognized by $\A$ and $\B$ in $\mathcal{C}$ respectively.

Consider union:  first we construct an automaton $\A + \B$ with $\varepsilon$-transitions that recognizes $R + S$.

\begin{center}
\begin{picture}(0,0)%
\includegraphics{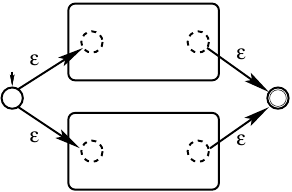}%
\end{picture}%
\setlength{\unitlength}{2072sp}%
\begingroup\makeatletter\ifx\SetFigFont\undefined%
\gdef\SetFigFont#1#2#3#4#5{%
  \reset@font\fontsize{#1}{#2pt}%
  \fontfamily{#3}\fontseries{#4}\fontshape{#5}%
  \selectfont}%
\fi\endgroup%
\begin{picture}(4420,2901)(221,-4414)
\put(2341,-2131){\makebox(0,0)[lb]{\smash{{\SetFigFont{10}{12.0}{\familydefault}{\mddefault}{\updefault}{\color[rgb]{0,0,0}$\A$}%
}}}}
\put(2341,-3841){\makebox(0,0)[lb]{\smash{{\SetFigFont{10}{12.0}{\familydefault}{\mddefault}{\updefault}{\color[rgb]{0,0,0}$\B$}%
}}}}
\end{picture}%
\end{center}

Define an $\varepsilon$-removal operation as follows. Consider  an $\varepsilon$-transition that goes from state $q_1$ to state $q_2$.  (We assume that $q_1 \not= q_2$.) We suppress the $\varepsilon$-transition and merge the two states $q_1$ and $q_2$ into one state $q$. Ascribe all incoming and outgoing transitions at $q_1$ and $q_2$ respectively, to the new state $q$.
The $\varepsilon$-removal is best visualized by pulling the state $q_2$ back to the state $q_1$, or by pushing the state $q_1$ forward to the state $q_2$ before actually merging them\footnote{Note that the result of the $\varepsilon$-removal operation does not depend on the orientation of the transitions.}.

\begin{center}
\includegraphics[scale=0.35]{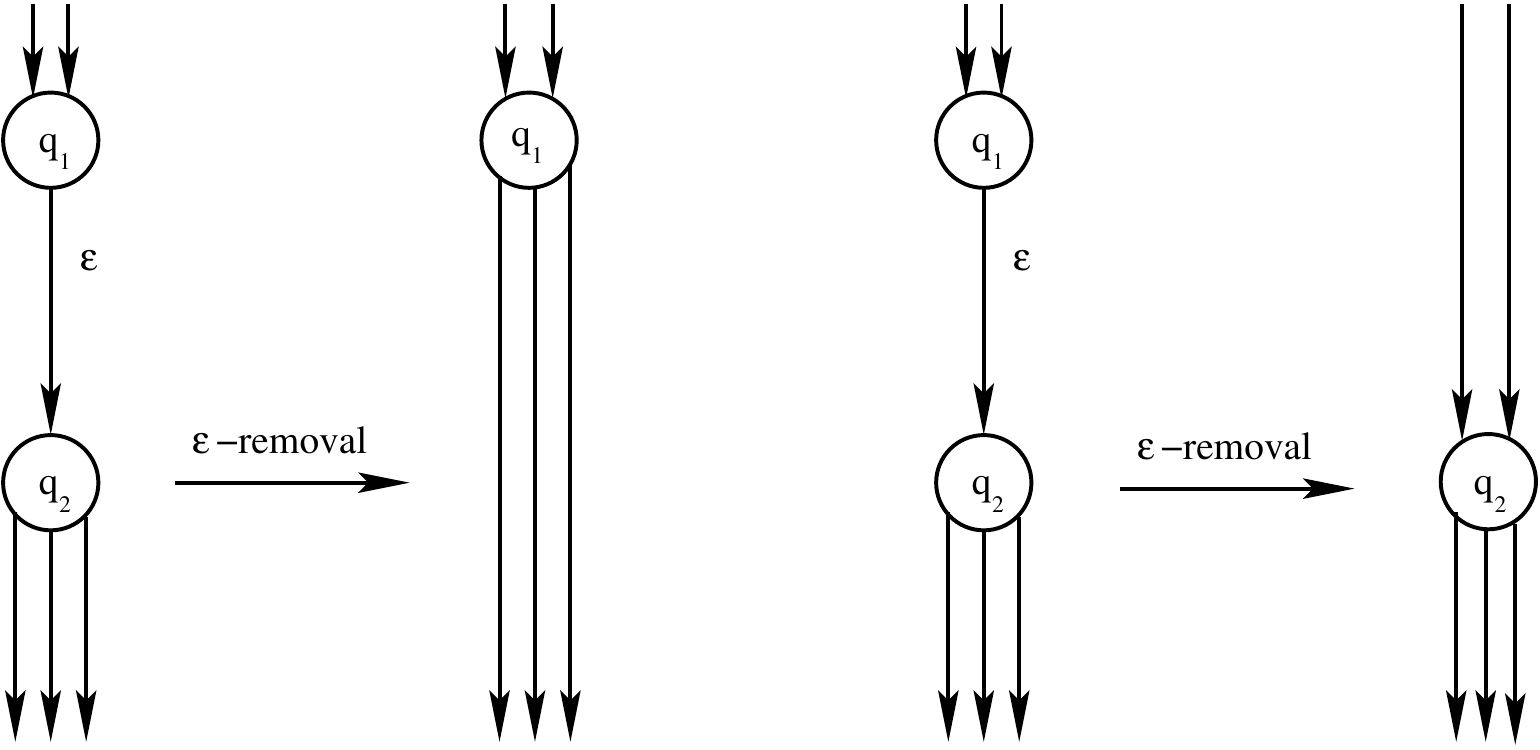}
\end{center}

We apply this operation four times (in any order) to the automaton above.
Clearly the result is an automaton that remains in~$\mathcal{C}$.

Consider concatenation: the following planar automaton with one $\varepsilon$-transition recognizes the expression $R\cdot S$.

\begin{center}
\begin{picture}(0,0)%
\includegraphics{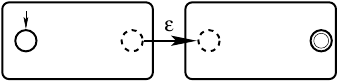}%
\end{picture}%
\setlength{\unitlength}{2072sp}%
\begingroup\makeatletter\ifx\SetFigFont\undefined%
\gdef\SetFigFont#1#2#3#4#5{%
  \reset@font\fontsize{#1}{#2pt}%
  \fontfamily{#3}\fontseries{#4}\fontshape{#5}%
  \selectfont}%
\fi\endgroup%
\begin{picture}(5151,1236)(1228,-2749)
\put(5131,-2176){\makebox(0,0)[lb]{\smash{{\SetFigFont{10}{12.0}{\familydefault}{\mddefault}{\updefault}{\color[rgb]{0,0,0}$\B$}%
}}}}
\put(2341,-2131){\makebox(0,0)[lb]{\smash{{\SetFigFont{10}{12.0}{\familydefault}{\mddefault}{\updefault}{\color[rgb]{0,0,0}$\A$}%
}}}}
\end{picture}%
\end{center}

Next we remove the $\varepsilon$-transition by the $\varepsilon$-removal operation. This provides us with the desired automaton in $\mathcal{C}$.

Finally consider Kleene's operation: suppose that the automaton $\A$ recognizes the expression $R$. The following planar automaton with three $\varepsilon$-transitions recognizes $R^{+} = \cup_{k \geq 1} R^k$.

\begin{center}
\begin{picture}(0,0)%
\includegraphics{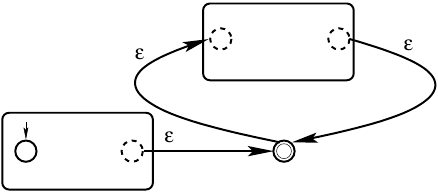}%
\end{picture}%
\setlength{\unitlength}{2072sp}%
\begingroup\makeatletter\ifx\SetFigFont\undefined%
\gdef\SetFigFont#1#2#3#4#5{%
  \reset@font\fontsize{#1}{#2pt}%
  \fontfamily{#3}\fontseries{#4}\fontshape{#5}%
  \selectfont}%
\fi\endgroup%
\begin{picture}(6665,2901)(1228,-2749)
\put(2341,-2131){\makebox(0,0)[lb]{\smash{{\SetFigFont{10}{12.0}{\familydefault}{\mddefault}{\updefault}{\color[rgb]{0,0,0}$\A$}%
}}}}
\put(5221,-466){\makebox(0,0)[lb]{\smash{{\SetFigFont{10}{12.0}{\familydefault}{\mddefault}{\updefault}{\color[rgb]{0,0,0}$\A$}%
}}}}
\end{picture}%
\end{center}

We remove the $\varepsilon$-transitions as before. This leaves us with the desired automaton in $\mathcal{C}$. This finishes the first proof.

The second proof is short but clever. Define $\A_n$ be the following deterministic finite automaton with set of states $[n] = \{1, \ldots, n\}$ and alphabet $A_n = \{ \sigma_{ij} \ | \ 1 \leq i,j \leq n\}$. For $1 \leq i,j \leq n$, set a transition with symbol $\sigma_{ij}$ from $i$ to $j$. We take $1$ to be the initial state and $2$ to be the final state.

Claim 1. {\emph{The automaton $\A_n$ has the universal property that any nondeterministic $n$-state automaton $\A = ([n], A, 1, \delta, 2)$ can be recovered (up to equivalence) by ``parallelization" of the transitions of $\A_n$. }}

Proof of the claim. Build an $n$-state automaton $\C_n$ by replacing each transition
\raisebox{-2mm}{\includegraphics[scale=0.8]{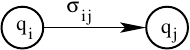}} in $\A_n$ by $T_{ij} = \{ \sigma \in A \ | \ j \in \delta(i, \sigma) \}$. (If $T_{ij}$ is empty, then we remove the transition $\sigma_{ij}$. Otherwise, we have $|T_{ij}|$ distinct ``parallel" transitions from $i$ to $j$.) The automaton $\C_n$ is equivalent to~$\A$. \hfill $\blacksquare$

Claim 2. {\emph{If $\A_n$ has an equivalent planar automaton $\B_n$ then any nondeterministic $n$-state automaton $\A = ([n], A, 1, \delta, 2)$ has an equivalent planar automaton.}}

Proof of the claim. We process the same proof as above with $\B_n$ instead of $\A_n$, observing that parallelization preserves planarity. \hfill $\blacksquare$

It remains to construct a planar automaton $\B_n$ equivalent to $\A_n$. The construction goes by induction on $n$. For $n = 3$, as a graph, $\A_3$ is the complete graph on three vertices, hence is planar. So we take $\B_3 = \A_3$. Suppose we have constructed a planar automaton $\B_n$, equivalent to $\A_n$, together with an embedding of $\B_n$ into $\mathbb{R}^{2}$ and a surjective map $\alpha:Q_n \to [n]$ from the set of states of $\B_n$ to the set of states of $\A_n$. We have to construct $\B_{n+1}$. Consider $\B_n \subset \mathbb{R}^2$. For any pair of distinct states $q,q'$ of $\B_n$, merge the transitions from $q$ to $q'$ and from $q'$ to $q$ into one unoriented edge. (If there is no transition, we do not perform any operation.) Finally we remove loops at each state. We obtain in this fashion an undirected simple graph $G_n$ whose vertices are exactly the states of $\B_n$. For each face $f$ of $\mathbb{R}^{2} - G_n$, place one vertex $v$ inside $f$ except for the exterior face (the unbounded component of $\mathbb{R}^2 - G_n$), and connect it to all vertices of the face $f$ and itself. We obtain a new graph $G_{n+1}$. See the figure below for the recursive constructive of $G_3, G_4$ and $G_5$.

\begin{center}
\includegraphics[scale=0.8]{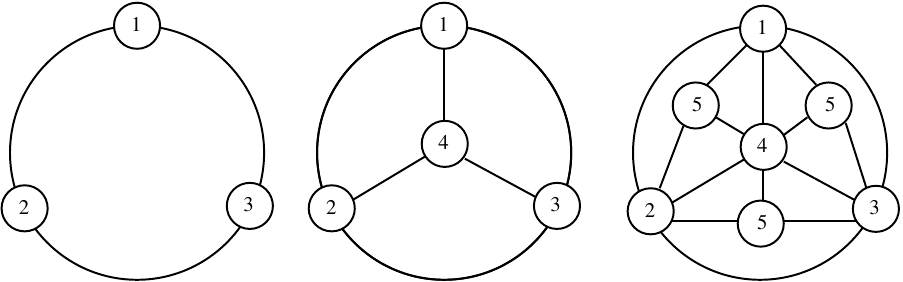}
\end{center}

We extend $\alpha$ by setting $\alpha(v) = n+1$. We restore all previous (oriented) transitions between any pair of vertices, we label the new loop at $v$ by the symbol $\sigma_{n+1,n+1}$ and we unfold each newly created edge from $v$ to any other (old) vertex $w$ into two transitions with opposite orientations with symbols $\sigma_{n+1,\alpha(w)}$ and $\sigma_{\alpha(w),n+1}$ respectively.

\begin{center}
\begin{picture}(0,0)%
\includegraphics{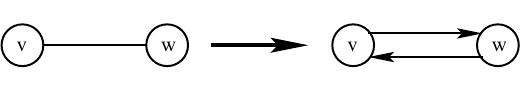}%
\end{picture}%
\setlength{\unitlength}{2072sp}%
\begingroup\makeatletter\ifx\SetFigFont\undefined%
\gdef\SetFigFont#1#2#3#4#5{%
  \reset@font\fontsize{#1}{#2pt}%
  \fontfamily{#3}\fontseries{#4}\fontshape{#5}%
  \selectfont}%
\fi\endgroup%
\begin{picture}(7927,1411)(-925,-2240)
\put(4951,-2131){\makebox(0,0)[lb]{\smash{{\SetFigFont{10}{12.0}{\familydefault}{\mddefault}{\updefault}{\color[rgb]{0,0,0}$\sigma_{\alpha(w),\alpha(v)}$}%
}}}}
\put(4951,-1096){\makebox(0,0)[lb]{\smash{{\SetFigFont{10}{12.0}{\familydefault}{\mddefault}{\updefault}{\color[rgb]{0,0,0}$\sigma_{\alpha(v),\alpha(w)}$}%
}}}}
\end{picture}%
\end{center}

This yields a new automaton $\B_{n+1}$. It is clear that the recursive step does not affect the initial state and the final state of $\B_{n+1}$ (that were already constructed together with $\B_3$). The automaton $\B_{n+1}$ is planar since $G_{n+1}$ is planar and the unfolding of the edges preserves planarity. It remains to see that $\B_{n+1}$ is equivalent to $\A_{n+1}$.
It follows from the definition of $\B_n$ that for $q, q' \in Q_n$ and $\sigma_{ij} \in A_n$, there is a $\sigma_{ij}$-transition from $q$ to $q'$ if and only if $\alpha(q) = i$ and $\alpha(q') = j$.
 It follows that every word recognized by $\B_{n}$ is also recognized by $\A_{n}$. To prove the converse, one shows that for any sequence $x_1 = 1, x_2, \ldots, x_k = 2$ in $[n]$ (which is a word in the language recognized by $\A_n$), there is a path\footnote{A path is a walk such that no edge occurs more than once and no internal vertex is repeated.} $y_1, \ldots, y_k$ in $\B_n$ such that $\alpha(y_j) = j$ for each $1 \leq j \leq k$. This is proved by induction on $k \leq m$ by using the facts that it is true for $m=3$ and that $G_{n}$ contains isomorphic copies of $G_{n-1}$. \hfill $\blacksquare$

\end{proof}

\section{Proof of the 1-gon Lemma} \label{sec:one-gon-lemma}

\subsection{Proof of Lemma \ref{lem:one-gon-is-bifacial}} \label{subsec:one-gon-is-bifacial}

Geometrically, a bifacial embedded loop is nothing else than a separating simple closed curve with a basepoint. It suffices to prove that a contractible simple closed curve is separating.
Consider an embedded loop $\alpha$ in $\Sigma^{1}$ based at $q \in \Sigma^{0}$. Assume that $\alpha$ is monofacial (nonseparating). Consider a small segment $I$ transversal (say, normal) to $\alpha$ such that $I \cap \alpha = \{ q \}$. Since $\alpha$ is monofacial, the endpoints of $I$ lie in the same connected component of $\Sigma - \A$. Hence $I$ extends to a loop $\beta$ such that $\beta \cap (\Sigma - \A) = \beta \cap \alpha = \{ q \}$. It follows that the algebraic $1$-homology intersection $[\beta] \cdot [\alpha] = \pm 1$. In particular, $[\alpha] \not= 0$ in $H_{1}(\Sigma)$. Thus $\alpha$ is not contractible. \hfill $\blacksquare$

\subsection{Proof of the 1-gon Lemma}

Consider a state $q$ of $\A \subset \Sigma$ that has at least one noncontractible loop.
Consider a small enough open disc $D$ in $\Sigma$ centered in $q$ such that the following properties hold:
1) $D \cap (\Sigma - \A)$ is a disjoint union of open cells; 2)
The intersection $D \cap \A$ is a wedge of semi-open arcs intersecting in
their common endpoint $q$; 3)
 Each arc $\alpha$ is bifacial: there are exactly two adjacent
cells $c, c' \in C = \{ c_1, \ldots, c_r \}$ such
 that $\alpha \subset {\rm{Fr}}(c) \cap {\rm{Fr}}(c')$.

Let $A$ be the set of arcs.  The orientation of $\Sigma$ induces a circular ordering $\alpha_1, c_1, \alpha_2,
 c_2, \ldots, \alpha_r, c_r$ of $A \cup D$ where the arcs and cells alternate and
such that any two consecutive cells are adjacent.

We fix now an arc $\alpha_1$ and perform successively the following operations on the arcs following the circular ordering. If the arc $\alpha_j$ does not belong to a loop (i.e. is part of a transition that is not a loop), we do not do anything.  Otherwise there is another arc $\beta$ belonging to the same loop. If the two arcs are enumerated consecutively in the circular ordering, we remove the whole loop inside $\Sigma$ and replace it by a small $1$-gon $\ell$ based at $q$ such that $\ell - q$ lies entirely in the open cell $c_{j}$.  At the end of the process, we have replaced all cycles of length $1$ by contractible loops, hence by $1$-gons. This does not change the surface hence it does not affect the genus of the embedding. In particular, if the embedding is minimal, the new embedding remains minimal (hence cellular), with the desired properties. Now by Lemma \ref{lem:one-gon-is-bifacial} each $1$-gon consists of one bifacial edge.
\hfill $\blacksquare$

\section{Proof of the Genus Formula} \label{sec:genus-formula-proof}

\subsection{Preliminary results} \label{subsec:preliminary}

Consider a minimal embedding of an automaton $\A$ into a closed oriented surface $\Sigma$. We let $e_{0}$ denote the number of $0$-cells (points, i.e. states), $e_1$ the number of $1$-cells (open transitions) and $e_2$ the number of $2$-cells (that is, the number of connected components of $\Sigma - \A$).
The first classical result is Euler's formula (\cite{Euler}, \cite[Chap. IV, \S 13]{Bredon}) that relates the genus to the CW-decomposition
of $\Sigma$. In our context, since $\Sigma$ is oriented and minimal, the formula takes the following form.

\begin{lemma}[Euler's formula]
\begin{equation}
\chi(\Sigma) = 2 - 2g(\A) = e_0 - e_1 + e_2. \label{eq:Euler-Poincare}
\end{equation}
\end{lemma}

Another useful observation is a consequence of the decomposition $\pi_{0}(\Sigma - \A) = \coprod_{k \geq 0} F_{k}$. Namely,
\begin{equation}
e_{2} = f_{1} + f_{2} + f_{3} + \ldots = \sum_{k \geq 0} f_{k}. \label{eq:two-cells}
\end{equation}
The sum above is finite since the total number of $2$-cells is finite. In particular, there is a maximal index $k \geq 0$ such that $f_{k} > 0$ and $f_{l} = 0$ for all $l > k$.

We need one more result that relates the number of $1$-cells to the number of faces.

\begin{lemma}\label{le:kfke}
\begin{equation}
2e_{1} = f_{1} + 2f_{2} + 3f_{3} + \ldots = \sum_{k \geq 0} k\ f_{k}. \label{eq:degree-formula}
\end{equation}
\end{lemma}

\begin{proof}
We begin with the relation  (\ref{eq:two-factor}): $\sum_{c \in \pi_{0}(\Sigma - \A)} \langle e,c \rangle = 2$. It follows that
$$ \sum_{e \in \Sigma^{1}} \sum_{c \in \pi_{0}(\Sigma - \A)} \langle e,c \rangle = 2 | \Sigma^{1} | = 2e_1.$$
Now use the decomposition of the cells into $k$-faces: $\pi_{0}(\Sigma - \A) =
\coprod_{k \geq 0} F_{k}$.
\begin{align*}
 \sum_{e \in \Sigma^{1}} \sum_{c \in \pi_{0}(\Sigma - \A)} \langle e,c \rangle = \sum_{c \in \pi_{0}(\Sigma - \A)} \sum_{e \in \Sigma^{1}} \langle e,c \rangle
 & = \sum_{k \geq 0}\ \sum_{c \in F_{k}} \sum_{e \in \Sigma^{1}} \langle e,c \rangle \\
 & = \sum_{k \geq 0}\ \sum_{c \in F_{k}} k \\
 & = \sum_{k \geq 0} k\ f_{k},
\end{align*}
where we used the relation $\sum_{e \in \Sigma^{1}} \langle e,c \rangle = k$ for a $k$-face $c$. This completes the proof. \hfill $\blacksquare$
\end{proof}

\subsection{Proof of Theorem \ref{th:genus-formula} (Genus formula)}

Consider a cellular embedding of $\A$ into a closed oriented surface $\Sigma$.  Euler's formula (\ref{eq:Euler-Poincare}) for the genus of $\Sigma$ gives $g_{\Sigma} = 1 - \frac{e_0 - e_1 + e_2}{2}$.
Since the automaton is complete, each state has exactly $m$ outgoing transitions. Therefore
$e_0 = e_1 / m$. Next use the relations (\ref{eq:degree-formula}) and (\ref{eq:two-cells}) to express $e_1$ and $e_2$ in terms of the $k$-faces. This yields the formula
$$ g_{\Sigma} = 1 + \sum_{k=1}^{+\infty} \frac{k(m-1)-2m}{4m} f_k. $$
Now $g(\A) \leq g_{\Sigma}$ with equality if and only if the embedding into $\Sigma$ is minimal. This achieves the proof.
 \hfill $\blacksquare$

\section{Proof of the Genus Growth Theorem} \label{sec:genus-growth-proof}
It is convenient to introduce the following functions:
 $$\displaystyle A(n)
=\sum_{k \geq 3} \frac{k(m-1)-2m}{4m}f_{k}(n)\ {\hbox{and}}\ B(n)
= \sum_{k \geq 3} k\ f_{k}(n).$$  We begin with
\begin{lemma} \label{lem:constante-alpha}
There is a constant $\alpha > 0$ such that
\begin{equation}
A(n) \geq \alpha B(n). \label{eq:prelim-equation}
\end{equation}
\end{lemma}

\begin{proof}
To prove the claim, we first find $\alpha
> 0$ such that
$$ \frac{(m-1)k - 2m}{4m} \geq \alpha\ k\ \ \ {\hbox{for all}}\ k \geq 3.$$
It suffices, therefore, to choose $\alpha$ such that
$$ \frac{m-1}{4m} - \frac{1}{2k} \geq \alpha \ \ \ {\hbox{for all}}\ k \geq 3.$$
This condition is satisfied if we choose $$ \underset{k \geq
3}{\inf} \left(\frac{m-1}{4m} - \frac{1}{2k}\right) =
\frac{m-3}{12m} = \alpha_0 \geq \alpha.$$ (Note that $\alpha_0 >
0$ for $m \geq 4$.)
This proves the lemma. \hfill $\blacksquare$ \end{proof}

\begin{lemma} \label{lem:rel-negligible-second}
\begin{equation}
\underset{n \to +\infty}{\lim} \frac{f_{j}(n)}{B(n)}= 0\ \ \ {\hbox{for}}\ j = 1,2.
\end{equation}
\end{lemma}

\begin{proof}
Since $f_k \leq z_k$, we have $\frac{f_k}{n} \leq \frac{z_k}{n} \underset{n \to +\infty}{\to} 0$ for $k = 1, 2$. Thus for any positive constants $a, b$,
\begin{equation}
 \frac{n}{a f_1 + b f_2} \underset{n \to +\infty}{\to} +\infty.
\label{eq:rel-negligible}
\end{equation}
Observe that
$ 2 \ m n = e_1(n) = f_1(n) + 2f_2(n) + B(n).$ Hence
$$ n = \frac{1}{2m}(f_1(n) + 2f_2(n) + B(n)).$$
Replacing $n$ in (\ref{eq:rel-negligible}) by this expression, with $a = 1/(2m)$ and $b = 1/m$, we find that
$$ \frac{\frac{1}{2m} B(n)}{\frac{1}{2m}f_1(n) + \frac{1}{m}f_{2}(m)} = \frac{B(n)}{f_1(n) + 2f_2(n)} \underset{n \to +\infty}{\to} +\infty.$$
Then $$\max \left(\frac{B(n)}{f_1(n)}, \frac{B(n)}{f_2(n)} \right) \geq \frac{B(n)}{f_1(n) + 2f_2(n)} \underset{n \to +\infty}{\to} \infty,$$  as desired. \hfill $\blacksquare$
\end{proof}

Let us come to the proof of the Genus Growth Theorem.
Let $\alpha > \varepsilon > 0$ satisfying the condition of Lemma \ref{lem:constante-alpha}.
Lemma \ref{lem:rel-negligible-second} ensures there is $N
> 0$ such that for any $n \geq N$,
$$ \left( \alpha + \frac{m+1}{4m}
\right)f_{1}(n) + \left( 2 \alpha + \frac{1}{2m} \right)f_{2}(n)
\leq \varepsilon\ B(n).$$ Hence for $n \geq N$,
$$ A(n) \geq (\alpha - \varepsilon) B(n) + \left( \alpha + \frac{m+1}{4m}
\right)f_{1}(n) + \left( 2 \alpha + \frac{1}{2m}
\right)f_{2}(n).$$ Thus
\begin{align*}
A(n) - \frac{m+1}{4m} f_{1}(n) - \frac{1}{2m} f_{2}(n) & \geq
(\alpha - \varepsilon)(B(n)+f_{1}(n) + 2f_{2}(n)) \\
& = 2(\alpha- \varepsilon)e_{1}(n) \\
& = 2(\alpha- \varepsilon)mn.
\end{align*}
According to Theorem \ref{th:genus-formula}, $A(n) -
\frac{m+1}{4m} f_{1}(n) - \frac{1}{2m} f_{2}(n) = g(n) - 1$. Thus
$$ g(n) \geq 1 +  2(\alpha - \varepsilon)m\ n.$$
This achieves the proof of the theorem. \hfill $\blacksquare$

\section{Proof of Theorem \ref{th:genus-lower-bound}}
\label{sec:genus-size}

Let the regular language $\Lang$ be represented by some minimal automaton $\A_{min} = \langle Q_{min}, A, q_{0,min}, F_{min}, \delta_{min}\rangle $ of size $n$. Let $\A = \langle Q, A, q_0, F, \delta \rangle$ be an automaton with minimal genus representing $\Lang$. According to Proposition~\ref{pr:reachable}, we can suppose without loss of generality that $\A$ is complete and that all states are reachable.  Suppose that $\A$ verifies the following properties
\begin{enumerate}[(i)]
\item $\A$ has at least $n$ states,
\item $\A$ contains no loops,
\item $\A$ contains no bigons.
\end{enumerate}

\noindent Then, 
\begin{eqnarray*}
g(\Lang)  = g(\A) &=& 1 + \sum_{k = 3}^{\infty} \cfrac{k(m-1) - 2m}{4m} f_k \hspace{10mm}\text{(by Genus Formula, (ii), (iii))}\\
&\geq & 1+ \left(\cfrac{m-3}{12m}\right)  \sum_{k = 3}^{\infty} k f_k \hspace{16mm}\text{(by Lemma~\ref{lem:constante-alpha})}\\
&\geq &1+\left(\cfrac{m-3}{6m}\right)e_1 \hspace{23mm} \text{ (by Lemma~\ref{le:kfke})}\\
&\geq&1+\left(\cfrac{m-3}{6}\right) n \hspace{26mm} \text{due to (i)}
\end{eqnarray*}

\noindent It remains to prove that $\A$ indeed satisfies the three properties (i-iii) listed above.
Following Myhill-Nerode, the set $Q_{min}$  of states of $\A_{min}$ is (isomorphic to) the set of equivalence classes for the non-distinguishing extension relation $R_\Lang$. Let the function $\rho: Q \to Q_{min}$ maps each state of $Q$ to its equivalent class in $Q_{min}$ as described by Proposition~\ref{pr:projection}.

\begin{enumerate}[(i)]
\item Since $\A_{min}$ is the minimal automaton, $\A$ has at least $n$ states.
\item Suppose that $\A$ contains a loop, that is $q = \delta(q,a)$ for some state $q$ and letter $a$. Then,  $\rho(q) = \rho(\delta(q,a)) = \delta_{min}(\rho(q),a)$. The second equality holds by Proposition~\ref{pr:projection}. This contradicts the hypothesis that $\A_{min}$ does not contain any loops.
\item In a similar way, suppose that $\A$ contains a bigon of the shape
\raisebox{-3mm}{\includegraphics[scale=0.7]{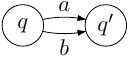}}. Then $$ \delta_{min}(\rho(q),a)  = \rho(\delta(q,a)) = \rho(\delta(q,b)) = \delta_{min}(\rho(q),b),$$ in contradiction with the hypothesis. If  $\A$ contains a bigon of the shape
\raisebox{-3mm}{\includegraphics[scale=0.7]{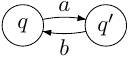}}. Then, $$\rho(\delta(\delta(q,a),b)) = \delta_{min}(\rho(\delta(q,a)),b) = \delta_{min}(\delta_{min}(\rho(q),a),b) = \rho(q),$$
which also contradicts the hypothesis.\hfill $\blacksquare$
\end{enumerate}

\section{Conclusion}

The topological tool we employ here, the genus as a complexity measure of the language,
leads to a viewpoint that seems orthogonal to the standard one: it is not compatible with set-theoretic
minimization (that is, state minimization).  However, the genus does behave similarly to the state
complexity with respect to operations such as determinization and union (up to a linear factor);
furthermore, there is a hierarchy of regular languages based on the genus. This suggests a more systematic
study of all operations : e.g. concatenation, star-operation, and composition of those. We take up this task
in a sequel to this paper \cite{BD}.

\bibliographystyle{alpha}
\bibliography{genre}

\noindent G. Bonfante, LORIA, Universit\'e de Lorraine, Campus scientifique, B.P.~239, 54506 Vandoeuvre-l\`es-Nancy Cedex

{\tt{guillaume.bonfante@loria.fr}}\\

\noindent F. Deloup, IMT, Universit\'e Paul Sabatier -- Toulouse III, 118 Route de Narbonne, 31069 Toulouse cedex 9

{\tt{florian.deloup@math.univ-toulouse.fr}}

\end{document}